\documentclass[sigconf]{acmart}
\AtBeginDocument{%
  \providecommand\BibTeX{{%
    \normalfont B\kern-0.5em{\scshape i\kern-0.25em b}\kern-0.8em\TeX}}}

\usepackage[utf8]{inputenc} 
\usepackage[T1]{fontenc}    
\usepackage{hyperref}       
\usepackage{url}            
\usepackage{booktabs}       
\usepackage{amsfonts}       
\usepackage{nicefrac}       
\usepackage{microtype}      
\usepackage{xcolor}         
\usepackage{caption}
\usepackage{subcaption}
\usepackage{enumitem}
\usepackage{natbib} 
\usepackage{mathtools} 
\usepackage{booktabs} 
\usepackage{tikz} 
\usepackage{amsthm}
\DeclareMathOperator*{\argmin}{arg\,min}
\usepackage{comment}
\usepackage{url}
\usepackage{kantlipsum}
\allowdisplaybreaks

\newtheorem{remark}{Remark}
\newtheorem{definition}{Definition}
\newtheorem{theorem}{Theorem}
\newtheorem{lemma}{Lemma}

\makeatletter
\newcommand\nocaption{%
    \renewcommand\p@subfigure{}
    \renewcommand\p@subtable{}
    \renewcommand\thesubfigure{\thefigure\alph{subfigure}}
    \renewcommand\thesubtable{\thetable\alph{subtable}}
}

\copyrightyear{2023} 
\acmYear{2023} 
\setcopyright{acmlicensed}\acmConference[FAccT '23]{2023 ACM Conference on Fairness, Accountability, and Transparency}{June 12--15, 2023}{Chicago, IL, USA}
\acmBooktitle{2023 ACM Conference on Fairness, Accountability, and Transparency (FAccT '23), June 12--15, 2023, Chicago, IL, USA}
\acmPrice{15.00}
\acmDOI{10.1145/3593013.3594117}
\acmISBN{979-8-4007-0192-4/23/06}

\begin{document}

\title{Detection and Mitigation of Algorithmic Bias via Predictive Parity}


\author{Cyrus DiCiccio}
\email{cjd48@cornell.edu}
\affiliation{%
  \country{USA}
}

\author{Brian Hsu}
\email{bhsu@linkedin.com}
\affiliation{%
  \institution{LinkedIn}
  \country{USA}
}

\author{Yinyin Yu}
\email{yinyyu@linkedin.com}
\affiliation{%
  \institution{LinkedIn}
  \country{USA}
}

\author{Preetam Nandy}
\email{preetamnandy@gmail.com}
\affiliation{%
  \country{Switzerland}
}

\author{Kinjal Basu}
\email{basukinjal@gmail.com}
\affiliation{%
  \country{USA}
}

\begin{abstract}  
Predictive parity (PP), also known as sufficiency, is a core definition of algorithmic fairness essentially stating that model outputs must have the same interpretation of expected outcomes regardless of group. Testing and satisfying PP is especially important in many settings where model scores are interpreted by humans or directly provide access to opportunity, such as healthcare or banking. Solutions for PP violations have primarily been studied through the lens of model calibration. However, we find that existing calibration-based tests and mitigation methods are designed for independent data, which is often not assumable in large-scale applications such as social media or medical testing. In this work, we address this issue by developing a statistically rigorous non-parametric regression based test for PP with dependent observations. We then apply our test to illustrate that PP testing can significantly vary under the two assumptions. Lastly, we provide a mitigation solution to provide a minimally-biased post-processing transformation function to achieve PP. 
\end{abstract}




\maketitle

\section{Introduction}
Fairness in artificial intelligence (AI) has been receiving much attention in recent years due to its tremendous impact on our everyday lives. Examples of AI-driven systems influencing our lives include social media \citep{garg2020social}, digital assistance \citep{maedche2019ai}, web searching \citep{zhang2018towards}, online and retail shopping \citep{pillai2020shopping}, healthcare \citep{panesar2019machine}, banking \citep{caron2019transformative}, and many more. Undoubtedly, as practitioners and researchers, it is our responsibility to make sure that the AI systems are fairly treating all individuals while making critical decisions such as the treatment plan of a patient \citep{ahmad2020fairness}, approval of the loan to a potential borrower \citep{mehrabi2021survey}, etc. An AI system can be unfair to one or more groups of individuals due to data bias that reflects biases present in society or the algorithmic bias resulting from the imperfection of AI systems in representing underlying data generating mechanisms. There are numerous definitions of group-fairness available in the literature including Demographic Parity, Equality of Opportunity, Equalized Odds and predictive parity \citep{verma2018fairness,barocas-hardt-narayanan}. This paper focuses on the detection and mitigation of algorithmic bias defined through predictive parity for large-scale AI systems.  

For binary classifiers, the predictive parity condition ensures that the observed outcome is conditionally independent of the protected attribute under consideration (e.g., race or gender) given the predicted outcome. In other words, we should observe the same (distribution of) outcome for the same prediction regardless of the protected attribute status. For example, an AI system predicting the chance of developing heart disease should have equal precision irrespective of the gender of a person. In the banking example, customers who got rejected for a loan should have an equal chance of being a defaulter regardless of their race.  In the binary classification setting, this condition is equivalent to equality of the means of the observed outcomes conditional on model score across groups, which naturally extends to the non-binary label setting. This condition is closely tied to the definition of sufficiency in the fairness literature \citep{barocas-hardt-narayanan}.  A closely related notion to predictive parity is that of marginal outcome testing \cite{ayres02}.  This line of work demonstrates that testing for equal outcomes in binary classification should be performed at the marginal threshold to avoid issues of ``infra-marginality'' encountered by methods aggregating all outcomes above the threshold.

Although the definition of predictive parity is well understood, to the best of our knowledge, there does not exist a formal statistical test of predictive parity in the presence of dependent data. Most large-scale AI applications contains multiple records from the same individual or user, and as a result, a natural dependency structure exists in the training data. In this paper, we propose statistically rigorous non-parametric regression based tests for predictive parity in the presence of dependent observations. We illustrate how Nadaraya-Watson (non-parametric) regression estimators can be used for testing, extend traditional asymptotic results for these estimators to handle the dependence, and demonstrate how these estimators can be used for post processing to ensure predictive parity across groups. 

Most of the fairness literature focuses on detection and mitigation of algorithmic bias in a single machine learning model. However, many large-scale industry applications uses a combination of models that are predicting different objectives and are ultimately combined through weighted summations to give a single score (such as those described in \cite{lifeed, ramanath2021efficient, chen2020,lada2021}). In this paper, we also extend our methods to test and mitigate for predictive parity in such multi-objective models. We devise methodology to simultaneously achieve predictive parity on each of the individual objective models as well as the overall model, which to our knowledge has not been previously addressed.


We also suggest approaches to test and mitigate discrepancies in marginal outcomes for binary classification problems which is closely connected to equality of conditional expected outcome.  Note that assessing disparity in marginal outcomes has been addressed by \cite{simoiu2017}, but this work requires the marginal threshold of a classifier be inferred. 
To summarize, we make three main contributions in this work: 
\begin{enumerate}[noitemsep,leftmargin=*]
    \item {\bf Testing for fairness in the presence of dependent data:} Collecting user data from engagement with machine learning systems inevitably leads to dependent observations in training data.  While it is standard practice in experimentation to account for dependencies in outcome variables, this issue has largely been ignored in fairness literature.  In this paper, we derive a formal statistical test for predictive parity which is valid under certain forms of dependence that are consistent with the practices seen in experimentation (e.g. as described in \cite{deng18}).  We also show that the computations used in testing can be re-used for post-processing mitigation.
    \item {\bf Predictive parity for multi-objective models:} Most of the fairness literature centers around single models predicting a single outcome (e.g. a logistic regression model predicting the likelihood of a binary outcome).  Existing methodology for calibration is insufficient to ensure fairness in multi-objective models, which are commonly used in large-scale applications.  In these cases, calibrating individual objective models does not guarantee fairness of the overall model.  We extend the framework of detection and mitigation of algorithmic bias from a single machine learning model to weighted combinations of models that is common practice in multi-objective settings.
    \item {\bf Connections to marginal outcome fairness for binary classification models:} We draw parallels between existing notions of marginal outcome fairness and calibration for classification algorithms.  In this context, we develop a rigorous approach to testing for fairness and propose a framework for mitigation.  
\end{enumerate}


The remainder of the paper is organized as follows. In Section \ref{sec:problem_statement} we provide a problem statement and precise formulation of predictive parity. In Section \ref{sec:testing} and \ref{sec:mitigation} we develop the formal testing and mitigation methodology respectively. We extend these methods to multi-objective models in Section \ref{sec:multi_objective} and draw connection to marginal outcome fairness in Section \ref{sec:infra-marginality}. Section \ref{sec:impact} discusses some of the metrics through which we can quantify the impact caused by deviating from predictive parity.
Finally, we present the experimental results in Section \ref{sec:empirical} before concluding with a discussion in Section \ref{sec:discussion}.  The supplement contains supporting examples, discussion, and proofs of the claims presented in the main text. We end this section with a brief overview of the related literature.

\textbf{Related Work:} While there are many definitions of fairness \cite{barocas-hardt-narayanan} which are often conflicting, we focus on predictive parity, which has also been formulated in the fairness literature as equal model calibration across groups \cite{Chouldechova17}. Calibration in predictive modeling has been well studied, see \citet{Kumar2019} and references therein for applications of calibration techniques in and outside of machine learning. In fairness, lack of equal calibration among groups is an issue in many contexts including criminal justice \cite{kleinbergFairnessTradeoff} and healthcare \cite{ObermeyerCalibrationInHealthcare}.

There have been multiple approaches to yield better calibrated models, \citet{HuangTutoricalCalibration} provides an overview of some such techniques, including  Platt scaling \cite{platt1999probabilistic}, Isotonic Regression \cite{zadrozny2002transforming}, and Bayesian Binning \cite{naeini2015obtaining}. More recent calibration methods include \citet{HebertJohnsonMulticalibration}, which focuses on subgroup calibration. Any post-processing calibration mentioned above can be used to achieve predictive parity by individually calibrating for protected groups. However, beyond applying a calibrator onto a model, there is also a question regarding the quality of calibration, which is the focus of this paper. Better methods of calibration testing have been researched in \citet{Kumar2019} and \citet{WidmannCalibrationTests}. These tests aim to assess the overall quality of calibration. Recently, \citet{tygert2022calibration} has proposed methods for assessing calibration differences between a subgroup and the full population. Our work is similar in that it also aims to study a similar issue of calibration differences across subgroups. However, the methodology we propose additionally accounts for testing with dependency between samples, which to our knowledge has not been studied before but is highly practical as repeat users appear in many large-scale internet applications.

\section{Problem Statement}
\label{sec:problem_statement}
Assume we observe data as $(Y,G,S)$, where $Y$ is an outcome being predicted by a machine learning model, $G$ is a group membership, and $S$ is a model score.  Note that $Y$ is often binary, but need not be for the remaining discussion.  Throughout, we assume that group membership is categorical.  We say that a score based machine learning model satisfies predictive parity when the following definition holds. 
\begin{definition}
A score based machine learning model satisfies \textbf{(conditional) predictive parity} for groups $g_1$ and $g_2$ if for each score $s$, $E(Y | G = g_1, S = s) = E(Y | G = g_2, S = s).$
\end{definition}
For a binary classifier, this condition asserts that for each possible value of the model score, the proportion of positive outcomes should be identical between the two groups which holds if, and only if, the outcome is independent of the group label when conditioned on the score.  

The definition
$
E(Y | G = g_1, S = s) = E(Y | G = g_2, S = s)$,
for all $s$ is a slight abuse of notation, particularly when conditioning on sets of measure zero. For this to be meaningful, it requires that by group, the scores and labels have a bivariate density that is nonzero on the same ranges for each attribute group, which may not always be the case.  A more precise formulation is to require that the random variables $E(Y | G = g_1, S)$ and $E(Y | G = g_2, S)$ defined on an underlying probability space $\Omega$ are equal {\it almost surely}. Note that the conditional expectations are not uniquely defined, and any random variable $Z$ that equals $E(Y | G = g_1, S)$ almost surely, is said to be a {\it version} of the conditional expectation (see Section 5.1 of \cite{durrett2019}).  The precise statement for predictive parity becomes, there exists versions of $E(Y | G = g_1, S)$ and $E(Y | G = g_2, S)$ which are equal almost surely.

However, it appears that this condition is not sufficient to ensure fairness. For instance, suppose a lending algorithm predicting the likelihood of default on a loan always scores men in the range $[0,1/2)$ and women in the range of $[1/2,1)$.  Because the scores are non-overlapping, predictive parity vacuously holds irrespective of the actual default rates, yet such an algorithm is clearly not fair in the intended sense.  Heuristically (and again with an abuse of notation), we want that  $E(Y | G = g_k, S = s) > E(Y | G = g_{k'}, S = s')$ whenever $s>s'$ and $k \neq k'$ to rule out the unpleasant behavior seen in the example.  To make this statement rigorous, let $O_g$ denote a version of conditional expectation, and let $S_g$ be a random variable distributed according to the score distribution of group $g$. 

Define
\begin{align*}
B(O_g, &O_{g'}) = \left\{ \omega \in \Omega :  O_g(\omega) > O_{g'}(\omega) \text{ and } S_g(\omega) < S_{g'}(\omega) \right\}
\end{align*}
to be the set on which scores are not fairly ordered with respect to score.  We are now ready to define predictive parity precisely. 

\begin{definition}\label{def:pred_rate_parity} A score based model satisfies {\bf (conditional) predictive parity} if $O_g$ equals $O_{g'}$ almost surely, and $P\left(B(O_g, O_{g'})\right)=0$ for any $O_g$ equals $O_{g'}$ which are versions of $E(Y | G = g, S)$ and $E(Y | G = g', S)$.  
\end{definition}

The mitigation approaches proposed in this work satisfy this formulation of predictive parity.

\section{Testing for predictive parity} 
\label{sec:testing}
Testing for predictive parity requires estimation of the conditional expected outcomes, $E(Y_i | G_i = g_1, S_i = s)$.  A natural estimator of these quantities is the {\it Nadaraya-Watson} estimator \cite{nadaraya64}, which is a class of non-parametric regression estimators.  Assume i.i.d. pairs $(S_i,Y_i)$ of model score and outcome label are observed.  These can be written as $Y_i = f(S_i) + \epsilon_i$, where $f(S_i) = E(Y_i | S_i)$, and $E(\epsilon_i | S_i) = 0$.  The Nadaraya-Watson estimator is
\[
\hat f(s) = \frac{\sum^n_{i=i} Y_i K_h(S_i - s)}{\sum^n_{i=i} K_h(S_i - s)}
\]
for a kernel $K_h(\cdot)$ having bandwidth parameter $h$.  We will further assume that the kernel takes the form $K_h(x) = K(x/h)$.  A common choice of kernel is the Gaussian kernel which specifies $K_h(x) = \exp \left(- x^2/(2h^2) \right).$  Note that for a fixed ``bandwidth'' parameter, $h$, the asymptotics of this estimator follow from a straight-forward application of the delta method (see Theorem 5.5.24 of \cite{casella2002}).  However, when the bandwidth is non-vanishing, the estimator is biased, and will not converge to $f(s)$.  In particular, we can decompose the usual pivotal quantity 
\[
\frac{\hat f (s) - f(s)}{\text{var} (\hat f(s))} = \frac{\hat f (s) - E(\hat f (s))}{\text{var} (\hat f(s))} + \frac{E(\hat f (s)) - f(s)}{\text{var} (\hat f (s))}.
\]

As such, the bandwidth needs to be tending to zero at a suitable rate, otherwise bias will be asymptotically non-zero and the term
$
\frac{E(\hat f (s)) - f(s)}{\text{var} (\hat f(s))}
$
will be non-vanishing.  Consequently, valid inference requires the bandwidth to tend to zero at an appropriate rate.  The asymptotics for these estimators with decreasing bandwidth are well studied when the data observed is independent and identically distributed, see for example \cite{mcmurry2008}.  However, we will derive the asymptotic distributions under certain types of dependence that are common to machine learning applications, typically arising from repeated observations from the same user of the platform.  This formulation for dependence, which is described in detail by \citet{deng18}, is standard practice in experimentation platforms used by technology companies.  

An appropriate model for such dependent data is that we observe tuplets $(Y_{m,i}, S_{m,i}, G_m)$ for $m = 1,...,M$, and $i = 1,...,n_i$.  Here, $m$ denotes the user (or individual), and $i$ denotes instances at which the user engages with the system.  We will assume that the $(Y_{m,i}, S_{m,i}, G_m)$ are identically distributed and that the observations are independent across users, but that $(Y_{m,i}, S_{m,i}, G_m)$ and $(Y_{m,i'}, S_{m,i'}, G_m)$ can be dependent.  For ease of exposition, we assume the observations for a given user are exchangeable.  In this setting, predictive parity can be defined as follows.

\begin{definition}\label{def:pred_rate_parity_dependent} A score based model satisfies {\bf (conditional) predictive parity} if $O_g$ equals $O_{g'}$ almost surely, and $P\left(B(O_g, O_{g'})\right)=0$ for any $O_g$ equals $O_{g'}$ which are versions of $E(Y_{1,m} | G_m = g_1, S_{1,m})$ and $E(Y_{1,m} | G_m = g_2, S_{1,m})$.  
\end{definition}

\begin{remark}
There are alternative formulations of predictive parity in the presence of dependent data, which are discussed in Section \ref{sec:alt_prp} of the appendix.
\end{remark}

 Evaluating this notion of predictive parity requires an estimator of the quantities $f_g(s) = E(Y_{m,1} | S_{m,1} = s, G_{m} = g)$.  We will use the Nadaraya-Watson style estimator
\[
\hat f_g(s) = \frac{\frac{1}{M}\sum_{m:G_m = g}\frac{1}{n_m}\sum_{i = 1}^{n_m} Y_{m,i} K((S_{m,i} - s)/h)}{\frac{1}{M} \sum_{m:G_m = g}\frac{1}{n_m} \sum_{i = 1}^{n_m} K((S_{m,i} - s)/h)}.
\]

We will provide a Theorem establishing the asymptotic normality of the Nadaraya-Watson estimator which characterizes the appropriate rate of convergence of the bandwidth to zero.

\subsection{Asymptotic Normality of the Nadaraya-Watson Estimator}

Before stating the result on the asymptotic normality, we will establish the requisite notation and state a sufficient set of assumptions on the underlying data generating mechanism and the choice of kernel.  

\subsubsection{Notation}  
Let $f_{S|g}(s)$, denote the density of $S_{1,1}$ conditional on $G=g$, and $f_{Y,S|g}(y,s)$ denote the joint density of $Y_{1,1}$ and $S_{1,1}$ conditional on $G=g$.  
Note that $f_g(s) = a_g(s) / b_g(s)$ where 
\begin{equation}
a_g(s)  = \int_{-\infty}^\infty y f_{Y,S|g}(y,s) d y \; \; \; \; \text{and} \; \; \; \; b_g(s) = \int_{-\infty}^\infty f_{Y,S|g}(y,s) d y~.
\end{equation}
Moreover, let $Y_{m,i,g} = Y_{m,i}\cdot I \left\{ G_{m,i} = g \right\}$ and its corresponding moments be,
\begin{align*}
f_g(S_{n,m};n_m) &= E \left( \left. Y_{m,i,g}\right|S_{n,m},n_m \right), \nonumber\\
\sigma^2_g(s;n_m) &= \text{Var} \left( \left. Y_{m,i,g}\right|S_{m, i, g} = s,n_m \right),\\
\rho_g(s_1, s_2;n_m) &= \text{Cov} \left( \left. Y_{m,i,g}, Y'_{m,j,g}\right|S_{m, i, g} = s_1, S_{m, j, g} = s_2,n_m \right).\nonumber
\end{align*}
The asymptotic distribution of the Nadaraya-Watson estimator additionally depends on, 
\begin{align}
\label{notation:fgsigma2}
\overline{f^2_g}(s) = E\left( \frac{1}{n_m}  f^2_g(s; n_m) \right)\; \text{ and }\;\; \overline{\sigma^2_g}(s) = E\left( \frac{1}{n_m} \sigma^2_g(s ;n_m) \right),    
\end{align}
and a few other terms which we push to the appendix for brevity.


\subsubsection{Assumptions}
\begin{enumerate}
    \item[A1] There exist constants $C$ and $\delta$ such that for all $s$, $$E\left( \left. |Y|^{2+\delta} \right| S = s \right) < C$$
    \item[A2] $\overline{\sigma^2_g}(s)$ and $\overline{f^2_g}(s)$  defined in \eqref{notation:fgsigma2} are continuous and bounded below by some constant $b>0$.
    \item[A3] $b_g(\cdot)$ has $d_1$ continuous derivatives and $f_g(\cdot)$ has $d_2$ continuous derivatives in a neighborhood of $s$.
    \item[A4] The observations for a given user are exchangeable given the number of observations generated by that user.  That is, for each $m$, $\left\{ (Y_{m,1}, S_{m,1}),..., (Y_{m,n_m}, S_{m,n_m}) \right\}$ is distributed as $\left\{ (Y_{m,\pi(1)}, S_{m,\pi(1)}),..., (Y_{m,\pi(n_m)}, S_{m,\pi(n_m)}) \right\}$ given $n_m$, where $\pi (\cdot)$ is a randomly chosen permutation of $\left\{ 1,...,n_m \right\}$.
    \item[A5] $(Y_{m,i} , S_{m,i})$ is independent of $(Y_{m',j} , S_{m',j})$ for $m \neq m'$.
    \item[B1] As $M \rightarrow \infty$, $h \rightarrow 0$ in such a way that $M \cdot h \rightarrow \infty$.
    \item[B2] $K(\cdot)$ integrates to one, the first $d_k \geq \max \left\{ d_1, d_2 \right\}$  moments of $K(\cdot)$ are zero, and $K(\cdot)$ has tails decaying faster than $x^{-m}$ for any $m$.
\end{enumerate}

Assumptions (A1)-(A3) are mild regularity conditions on the smoothness of the conditional expected outcomes, and variability of the outcome labels required for asymptotic normality.  Assumption (A4) is made for notational compactness (modifications when these conditions are not met required are noted in the proofs).  A similar result holds when the data are not exhangeable and the score and labels are allowed to have arbitrary dependence, with the estimand and asymptotic variance modified to depend on the full joint distribution of the observed data.  Assumption (A5) asserts that the scores and outcomes for two users are independent, which is akin to the standard iid assumptions in non-repeated measurements results.  Assumptions (B1) and (B2) pertain to the user driven choice of kernel and bandwidth.  Provided the quantities appearing in condition (A3) have at least one continuous derivative, many common choices including the Gaussian and Epanechnikov kernels satisfy (B2).  Appropriate choice of $h$ satisfying (B1) is characterized in the following Theorem. 

\begin{theorem}\label{thm:asymptotic_normality}
Under the Assumptions (A1)-(A5), (B1) and (B2), there exists a function $\sigma^2_g (\cdot)$ such that 
\[
\sqrt{Mh} \left( \hat f_g(s) - f_g(s) + o(h^d) \right) \rightarrow N\left( 0, \sigma^2_g (s) \right)
\]
where $d = \min \left\{ d_1,d_2 \right\}$ (defined in the Appendix).  In particular, choosing $h = O(M^{-1/(2d + 1)})$,
we have that 
\[
\sqrt{Mh} \left( \hat f_g(s) - f_g(s) \right) \rightarrow N\left( 0, \sigma^2_g (s) \right).
\]
\end{theorem}

Testing for predictive parity can be accomplished by testing equality of the conditional expectations at a fixed range of score values and simply applying a Bonferroni (or any other family-wise error rate controlling \cite{benjamini1995controlling}) correction to the $p$-values computed using the result of Theorem \ref{thm:asymptotic_normality}.  With additional effort, methods such as those in \citet{luedtke2019} or \citet{srihera2010} for testing equality of unknown functions can be modified to account for dependence and provide higher-powered tests for predictive parity.

\section{Mitigation Approaches}
\label{sec:mitigation}
Writing $f_g(s) = E(Y|S=s, G = g)$, we could, for members of group $g_1$, replace a score $s$ with the score $\tilde s$ satisfying $f_{g_1}(\tilde s) = f_{g_2}( s)$, or equivalently, transform the scores of both groups such that $f_{g_1}(\tilde s) = f_{g_2}( \tilde s)$.  Any such transformation would satisfy predictive parity.  Arguably the simplest, and most intuitive approach is to ``calibrate'' the scores, i.e. to apply a transformation $\tilde S = t(S,G)$ such that 
\[
E(Y | \tilde S = s, G = g) = s
\]
for all $g$ and $s$.  While there are numerous such transformations, the ideal transformation in the mean-squared error sense is
\[
E(Y|S, G) = \argmin_f E \left( Y - f(G, S) \right)^2.
\]
\begin{remark}
Note that after calibrating model scores for each group under consideration, the scores may no longer have common support across groups, motivating the need for Definition \ref{def:pred_rate_parity}.  
\end{remark}
There are numerous approaches for achieving calibration including isotonic regression, Platt's scaling, and binning.  Platt's scaling and isotonic regression have been demonstrated to have good empirical properties, though may fail to asymptotically yield the ideal transformation in cases where the requisite model assumptions are not met.  Binning can introduce unintended bias. We thus suggest Nadaraya-Watson estimators as an alternative to these methods in the next section.

\subsection{Non-Parametric Calibration}
It was seen in Section \ref{sec:testing} that the Nadaraya-Watson estimator provides a consistent estimate of the expected outcome at a given score for suitably chosen kernels and bandwidth parameters.  Let $\hat f_g (\cdot)$ be such an estimator for group $g$.  Replacing the score with the estimated conditional expected outcome given the original score (asymptotically) calibrates the model.  Of course, in real-world applications, one can evaluate the conditional expectation at a fixed set of score values and linearly interpolate between these values as follows.  Consider a binning (or partition) $B_1,...,B_K$ of the score space with $B_i = [l_i, u_i)$.  For a given score, denote by $B(s)$ the bin $B_i$ containing $s$, i.e. for which $s \in B_i$. A simple linear interpolation score transformation is given by
\[
t(s,g) = \frac{\hat f_g(u_k) - \hat f_g(l_k)}{u_k - l_k} \cdot \left( s - l_k \right) + \hat f_g(l_k)
\]
where $k$ is the index such that $s \in B(k)$.  This transformation is asymptotically unbiased at the endpoints of the bins in the sense that the transformation converges to the ``ideal'' transformation, i.e. $t(s,g) \rightarrow f_g(s)$ almost surely for scores on the endpoints of the bins.  Provided the binning is reasonably fine, the transformation will have small bias over the ideal transformation.  While this approach generally requires more data than other methods of calibration, it has several advantages.  This approach is robust against the monotonicity assumption required by isotonic regression.  This method also does not require the logistic model specification used by Platt's scaling and is applicable to non-categorical outcome data, which cannot be said of Platt's scaling.  A further discussion of these approaches is given in Appendix \ref{sec:calibration_appraoches}.

\section{Multi-Objective Fairness}\label{sec:multi_objective}
Many online recommendation systems including those used by large technology companies such as LinkedIn \cite{lifeed}, Meta \cite{lada2021}, and Pinterest \cite{chen2020}, combine models of various outcomes into a single score through a weighted average.  Suppose that scores $S^{(1)}_i,...,S^{(K)}_i$ predict outcomes $Y^{(1)}_i,...,Y^{(K)}_i$, and that ranking, recommendation, classification, etc. is then based on a composite score 
$S_i = \sum_{k=1}^K w_k S^{(k)}_i$
where the $w_k$ are weights chosen to balance the trade-offs between the various objectives (for instance, models of various engagement actions such as likes, clicks, or comments may be combined in this way to rank news feed items).  
In such cases, it may be of interest to not only have calibration of the individual models, but to have calibration of the composite score $S_i$ as a predictor of the composite outcome 
$Y_i = \sum_{k=1}^K w_k Y^{(k)}_i$.  In general, calibration of the individual model is not enough to guarantee calibration of the composite model; however, a slightly stronger condition is sufficient.

\begin{lemma}
\label{lemma:multi-fairness}
Suppose that $E\left( Y^{(k)}_i | S^{(1)}_i = s_1,...,S^{(K)}_i = s_K \right) = s_k$
for all $k$ and all $s_1,...,s_K$, then 
$E\left( Y_i | S_i = s\right) = s_i$.
\end{lemma}

Interestingly, if the individual models are calibrated, then there can exist an individual model for which $S^{(k)}_i$ is not the best linear predictor of $Y^{(k)}_i$ given $S^{(1)}_i,...,S^{(K)}_i$.  It follows that there can be some other set of weights $\tilde w_1,...,\tilde w_K$ such that 
$
\tilde S_i = \sum_{k=1}^K \tilde w_k Y^{(k)}_i
$
is a better (in terms of mean-squared error) predictor of $Y_i$ than $S_i$.  A simple example of this is given in Appendix \ref{sec:mo_example}.

\subsection{Multi-Objective Mitigation}
Similarly to the univariate case, mitigation in the multi-objective setting can be achieved by replacing the original scores with appropriate conditional expectations.

\begin{theorem}\label{thm:multi-fairness}
Suppose we observe scores $S_k = s_k$
Defined transformed scores as 
\[
\tilde S^{(k)} = E(Y^{(k)} | S^{(1)},...,S^{(K)}, G)
\]
for $k = 1,...,K$.  That is, replace observe scores $S^{(1)} = s_1,...,S_{(K)} = s_K$, with transformed scores $\tilde S_k = E(Y^{(k)} | S^{(1)} = s_1,...,S^{(K)}=s_K, G)$.  Then, the transformed composite score 
\[
\tilde S = \sum_k w_k \cdot \tilde S^{(k)}
\]
satisfies predictive parity for the multi-objective outcome in the sense that
\begin{align*}
E\bigg(\sum_k w_k &\cdot Y^{(k)}  | \sum_k w_k \cdot \tilde S^{(k)} = s, G = g_1 \bigg) \\
& = E\bigg(\sum_k w_k \cdot Y^{(k)} | \sum_k w_k \cdot \tilde S^{(k)} = s, G = g_2\bigg)
\end{align*}

for all $s$ and weights $w_1,...,w_K$, while also maintaining predictive parity on the individual models in the sense that 
\[
E( Y^{(k)} | \tilde S^{(k)} = s_k, G = g_1) = E( Y^{(k)} | \tilde S^{(k)} = s_k, G = g_2) 
\]
for all $s$ and $k$.
\end{theorem}

In particular, if we calibrate the individual models conditional on (all of) the scores, then the composite score will satisfy predictive parity regardless of the weights chosen.  In many internet industry applications, the weights are chosen through online experimentation (for instance, see \cite{Agarwal:oms}), and this allows for guarantees of predictive parity regardless of the weights ultimately chosen.  Even if calibration of the composite outcome is not of interest to the practitioner, the proposed transformations to the individual model scores improves accuracy in the mean squared error sense, i.e.
\begin{align*}
E(Y^k &| S_1,...,S_K, G)  = \argmin_t E\left( Y^k  - t(S_1,...,S_K, G)\right)^2.
\end{align*}

\begin{remark}
Mitigation requires estimating conditional expectations.  While the Nadaraya-Watson have obvious extensions to the multivariate context, they suffer from the ``curse of dimensionality.'' The rate of convergence of these estimators becomes $\sqrt{M\cdot h^K}$ in cases where $d=K$.  While this may not be a serious concern in large data applications, especially when $K$ is of modest size, it can be challenging in other settings.  In small data settings and when each of the outcome labels are binary, Platt's scaling can be modified for this task by fitting a logistic regression of each of the outcome variables on the full set of model scores.    
\end{remark}

\section{Fairness for Classifiers and Connections to Marginal Outcome Fairness}
\label{sec:infra-marginality}

In classification settings, candidates whose scores are above (or below) a common threshold receive the same treatment. For instance, if an algorithm predicting default risk is used to approve loans, all individuals with a suitably low score are given the same treatment in the sense that they are approved for a loan. In these settings, the matter of fairness comes down to understanding whether an equitable standard is applied to all groups. While all groups are subjected to a common score threshold, mis-calibration in scores between groups can result in different effective thresholds being applied to different group as measured by the marginal outcome of each group \citep{corbett2018measure}. Marginal outcomes have been used to detect discrimination in applications such as police search \citep{10.1257/000282806776157579} and medical treatment \citep{NBERw16888}.  \citet{ayres02} notes in such applications of ``outcome tests'' which compare the average (or infra-marginal) outcomes stemming from classification decisions across groups, can give misleading impressions of algorithmic bias and argues for the need to assess disparities in outcomes associated with marginal (or threshold) decisions.  
Formally, suppose that a classifier yields a positive decision if the model score $S$ exceeds a marginal threshold $t$.  Fairness in marginal outcomes can be characterized by the following definition.

\begin{definition}
A classifier resulting in a positive classification exactly when a model score $S$ exceeds a threshold $t$ is {\bf marginal outcome fair} if $E \left( Y_i |G_i = g_1, S_i = t \right) = E \left( Y_i |G_i = g_2, S_i = t \right)$.
\end{definition}

We now describe the equivalence between marginal outcome fairness and predictive parity.  If the model satisfies predictive parity in the sense of definition \ref{def:pred_rate_parity}, then the marginal outcome fairness condition is satisfied.  On the other hand, if this condition is met, calibrating the model at score values other than the marginal threshold will not affect the classifications decisions provided the conditional outcomes are monotonic.  Of course, this equivalence requires that the average outcome at all scores below (or above) the marginal threshold does not exceed (or fall below) the average outcome at the marginal threshold and so it may also be of interest to ensure that the outcomes of candidates below the marginal threshold are indeed below that of the marginal threshold. However, without observed data on such observations, evaluating this monotonicity may require strategic randomization of classifications, which is beyond the scope of the present work.  

The methodology derived for testing predictive parity can be applied to this context.  A consistent and minimally biased estimator of the marginal outcome can be computed as 
\[
\hat f_g(t^*) = \frac{\frac{1}{M}\sum_{m:G_m = g}\frac{1}{n_m}\sum_{i: S_i \geq t^*} Y_{m,i} K((S_{m,i} - s)/h)}{\frac{1}{M} \sum_{m:G_m = g}\frac{1}{n_m} \sum_{i: S_i \geq t^*} K((S_{m,i} - s)/h)}
\]
Note that this has the minor modification of only averaging scores above the threshold, since it is now assumed that we are missing labels below the marginal threshold.  Asymptotic normality of the estimated average marginal outcome holds analogously to Theorem \ref{thm:asymptotic_normality}, allowing for rigorous testing of equality between groups.   

\begin{theorem}
\label{thm:inframarginality}
Under the assumptions of Theorem \ref{thm:asymptotic_normality}, there exists a function $\sigma^2_g (\cdot)$ such that 
\[
\sqrt{Mh} \left( \hat f_g(t^*) - E \left( Y_i |G_i = g, S_i = t^* \right) \right) \rightarrow N \left( 0, \sigma^2_g (t^*) \right).
\]
\end{theorem}

\subsection{Mitigation} When the test rejects, there is evidence that one group is ``advantaged'' in the sense that the marginal candidates from this group are held to a higher standard than the marginal candidates from the other,``disadvantaged'' group.  In order to correct for differences in treatment of marginal (or threshold) candidates, we can raise the threshold applied to the advantaged group, lower the threshold applied to the disadvantaged group, or some combination of both.  Raising the threshold of the advantaged group may be preferable from a technical standpoint since it avoids the need to infer the outcomes below the marginal threshold, however, it may not be the best choice for the problem at hand (e.g. the practitioner may be satisfied with the threshold applied to the advantaged group and might rather decrease the threshold for the disadvantaged group).

Let $\hat f_g (t)$ be a prediction of the average outcome of group $g$ at a threshold $t$.  For $t$ above the marginal threshold, this can simply be the non-parametric regression estimates described in Section \ref{sec:testing}.  Below the marginal threshold, the average outcome can be estimated, for instance, by a linear regression of outcomes near the marginal threshold.  Now, the marginal threshold can be replaced with group-specific thresholds $t^*_g$ which satisfy 
\begin{equation}\label{eqn:marginal_outcome}
\hat f_{g_1} (t^*_{g_1}) = \hat f_{g_2} (t^*_{g_2}).
\end{equation}

\begin{remark}
While it may seem unfair to apply differing thresholds to each group, this procedure is equivalent to calibrating models and applying the same threshold to each group.  Here, the differing thresholds are to account for the differing calibration of the models.
\end{remark}

Note that there are typically many thresholds yielding equal marginal outcomes, and additional constraints can be added to choose appropriate thresholds.  For instance, thresholds can be modified to mitigate marginal outcome bias while maintaining the same overall rate of positive classification.  Let $\hat F_{g_i}(\cdot)$ be the cumulative distribution function of the scores arising from group $g_i$.  Choosing fair group-level thresholds maintaining the same overall positive classification rate can be accomplished by adding the constraint that 
\begin{equation} \label{eqn:budget}
\sum_i p_{g_i} \int_0^{t^*_{g_i}} \hat f_{g_1} (t) \hat F_{g_i}(t) =  \sum_i p_{g_i} \int_0^{t^*} \hat f_{g_1} (t) \hat F_{g_i}(t) 
\end{equation}
where $p_g = P(G = g)$. 
\begin{theorem}
\label{thm:infra2}
Adjusted thresholds $t^*_{g}$ that mitigate against marginal outcome bias while maintaining the overall positive classification rate, i.e. that solves \eqref{eqn:marginal_outcome} and \eqref{eqn:budget} exist when the average outcome functions and score distributions are continuous. 
\end{theorem}

\begin{remark}
In applications without observed outcomes below the marginal threshold, it is important to note that the predicted outcomes below the marginal threshold will be merely suggestive.  In practice, one can either either iteratively correct the thresholds and re-test for marginal outcome fairness or temporarily lower the threshold of the disadvantaged group to learn the appropriate threshold.
\end{remark}

\section{Quantifying Impact}
\label{sec:impact}
While departures from predictive parity (or correspondingly difference in calibration curves) can be indicative of a fairness issue, they do little to quantify the tangible disadvantages that groups face due to unfair calibration.  For instance, in ranking contexts, it is not necessarily the case that calibrating unfairly calibrated models would influence the rankings, so it can be hard to interpret the true impact of calibration.  In this section, we describe metrics to quantify the impact of deviations from predictive parity. 

\subsection{Expected Calibration Error}


The definition of predictive parity is closely related to notions of calibration. It is common to quantify the degree of mis-calibration through the discrepancy between the score and the conditional estimated outcome. The True Calibration Error metric (see, for instance, \citet{roelofs2022mitigating} is defined as:
\begin{equation}
    \text{TCE} = \left( E \left( s - E(Y|s) \right)^2 \right)^{1/2} \label{eq: TCE}
\end{equation}
Although this metric does not directly correspond to predictive parity, disparities in this metric across groups can be indicative of unfair model performance. This has been measured in several previous works through the equal-width Expected Calibration Error \cite{guoCalibration}, which computes $E(Y|s)$ through binning scores and computed per-bin expected outcomes $Y$. \citet{roelofs2022mitigating} and \citet{Kumar2019} found that this binning mechanism leads to a biased estimate of TCE and develop other estimators; however, they note that the estimators may not have desirable properties such as consistency (that is, the estimators may not asymptotically tend to the true TCE). Using Nadaraya-Watson estimators, the TCE can be consistently estimated by
\[
\widehat{\text{TCE}} := \frac{1}{n}\sum_i \left(s_i - \hat f(s_i) \right)^2~.
\]
We demonstrate a version of this in Section \ref{sec:empirical}. Further details are given in Appendix \ref{appendix: ECE}.

\subsection{Impact in Classification and Ranking Settings}

While unfairness in predictive parity can be indicative of model bias, the definition does not provide insights into the tangible effects in classification or recommender systems.  In such applications, it can be helpful to quantify the discrepancy in classification decisions or ranks between a model and the optimally calibrated version of the same model. Specifically, if a model score $S$ predicting an outcome $Y$, then it is of interested to compare either the classification decisions or ranks obtained by using $S$ with the optimally calibrated version (by group $G$),
\[
S^* = E \left( Y | S, G \right)~.
\]

For a classifier predicting a positive label whenever $S$ exceeds a threshold $t$, the delta in positive classification-rate due to (unfair) calibration for a group $g$ can be defined as
\[
\Delta PCC_g = E \left( I \left\{  S^*_i > t\right\} | G = g \right) - E \left( I \left\{  S_i > t\right\} | G = g \right)
\]
(or equivalently by modifying the threshold rather than the score itself as discussed in Section \ref{sec:infra-marginality}).  This quantity can be consistently estimated by an empirical version using the Nadaraya-Watson estimator of $S^*$
\begin{theorem}\label{thm:delta_classification}
Under the conditions of Theorem \ref{thm:asymptotic_normality},
\[
\widehat{\Delta PCC}_g = \frac{1}{n_g}\sum_{i:G_i = g} I \left\{  S_i > t\right\}  - I \left\{  \hat f_g(S_i) > t\right\} \xrightarrow{p}  \Delta PCC_g
\]
(where $\xrightarrow{p}$ denotes convergence in probability).
\end{theorem}

Similarly, in ranking contexts where we are considering fairness to the recommended people (as in connection recommendation products) or content creators of the ranked items, it is useful to quantify the disparity in rank positions.  Suppose that for each query $q$, a model produces scores $S_{1,q},...S_{m,q}$ for $m$ items or individuals, and that these items are ranked according to the scores.  Let $R(S_{i,q}, q)$ denote the ranking function of the items which takes value $r$ when $S_{i,q}$ is the $r$-th largest score among the scored items for query $q$.  We define the delta in rank due to (unfair) calibration for group $g$ as
\[
\Delta RC_g = E \left( R(S_{i,q}, q) -  R(S^*_{i,q}, q) | G_i = g \right)~.
\]

\begin{theorem}\label{thm:delta_rank}
Under the conditions of Theorem \ref{thm:asymptotic_normality},
\[
\widehat{\Delta RC}_g = \sum_q \frac{1}{m} \sum_{i=1}^m I \left\{ G_i = g \right\} \cdot \left( R(S_{i,q},q) - R(\hat f_g(S_{i,q}),q) \right) \xrightarrow{p} \Delta RC_g~.
\]
(where $\xrightarrow{p}$ denotes convergence in probability).
\end{theorem}

This metric quantifies the aggregate loss (or gain) of exposure that a group receives as a result of an unfairly calibrated model.  Modifications of this metric, such as considering the change in rank per slot, or considering only the top few slots may also be of interest.  Once again, the Nadaraya-Watson estimators can be used to obtain a consistent estimator of this metric.  

\begin{remark}
Confidence intervals can be obtained for either of these metrics, but additional care must be taken because the non-parametric rates of convergence of the Nadaraya-Watson estimators which are slower then the usual $\sqrt{n}$ rates.  We omit these results for brevity.  
\end{remark}

\section{Empirical Findings}
\label{sec:empirical}
We now empirically validate the effectiveness of the non-parametric test in two ways. First, we illustrate the importance of explicitly distinguishing between testing under dependence and independence by showing that changing these assumptions can lead to vastly different conclusions about predictive parity group fairness. Secondly, we will compare the non-parametric test to existing calibration errors to show that it has much lower bias when testing under dependence and performs just as well under independence. 

We will rely on the Bias-By-Construction (BBC) framework proposed in Roelofs et. al . \cite{roelofs2022mitigating} for these experiments. In the BBC framework, the score distribution is parameterized by a Beta distribution $S\sim Beta(\alpha,\beta)$ while the conditional expectations given a score is a parametric function of score $E(Y|s)=g^{-1}(b_{0}+b_{1}f(s))$. The advantage of this framework is that it enables calculation of the true calibration error TCE \eqref{eq: TCE} by computationally evaluating it across the range of scores $s$ as the function $E(Y|s)$ is known by construction. Similarly, the empirical calibration error measured by function $f$, $CE_{f}$ can be computed by simulating outcome draws $Y(s)\sim ~Bernoulli(E(Y|s))$ for given scores. The bias is then represented as the difference between the simulated and true error $Bias(f)=CE_{f}-TCE$. 

First, we show the potential impacts of dependency assumptions on outcomes testing. In this experiment, we generate 50k samples for each of two groups from the same $Beta(\alpha, \beta)$ distribution. However, we slightly change the constants of the conditional expectation from one group to another. The resulting calibration plots are shown in the left plot of Figure \ref{fig:calibrationPlot}. Now, at each of the 10 score deciles of group 1, we inject 1000 samples with higher expectations\footnote{Specifically, we increase the expected outcome of these 10 individuals by $0.7s$ such that the expectation increases more for higher scoring individuals.}. This introduces 10k more total samples, but importantly there are only 10 more unique individuals in the new data (less than 1\% more unique individuals). As the figure shows, this can drastically change the readout of the expected outcomes curve. Whereas group 2 was previously advantaged by the miscalibration\footnote{By "advantaged by miscalibration," we refer to the fact that conditional on score, group 1 has higher expected outcome, but someone interpreting the score would assume that they have the same outcome. This can influence how individuals are ranked or perceived if the scores have downstream uses}, with just 10 more individuals, the curve now shows that group 1 is advantaged. This is captured by the non-parametric test as well. In Figure \ref{fig:IndVsAgg} we see that when evaluating expectations under dependence, we recover the original conditional expectation pattern with group 1 have lower conditional expectations than group two. However, when evaluating with independent samples, the conditional expectation of group 1 is mostly higher than that of group 2. These examples show the importance of accounting for sample dependence, as merely a few individuals with a large number repeated samples can vastly skew conclusions.

\begin{figure*}[!ht]
     \nocaption
     \centering
     \begin{subfigure}[b!]{0.48\textwidth}
         \centering
         \includegraphics[height=5.5cm]{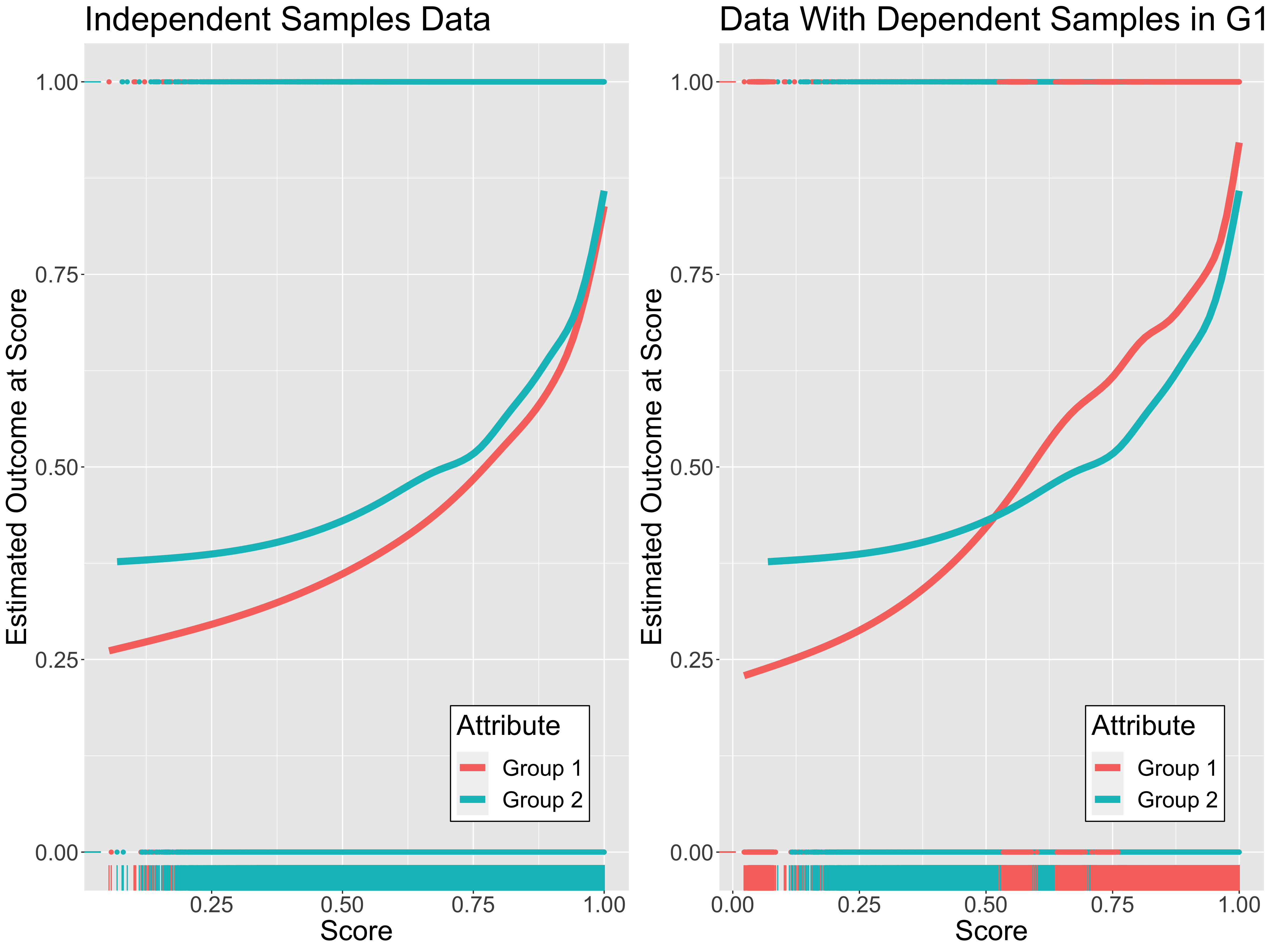}
         \caption{Plots of smoothed conditional expectations when we only have independent samples (left) versus if we inject 10 new individuals each of whom have many repeated samples that are skewed upwards in conditional expectation.}
         \label{fig:calibrationPlot}
     \end{subfigure}
     \hfill
     \begin{subfigure}[b!]{0.48\textwidth}
         \centering
         \includegraphics[height=5.5cm]{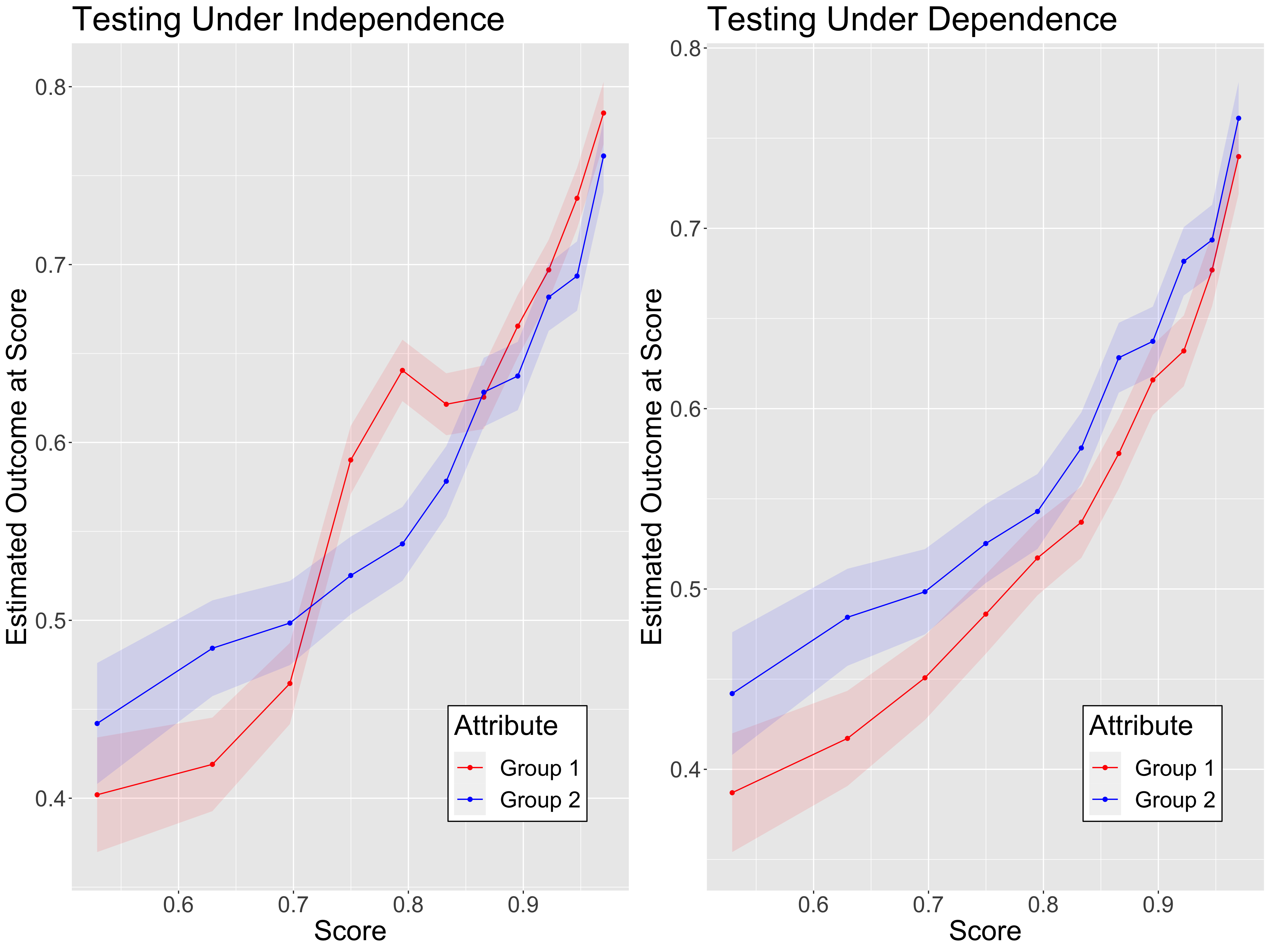}
         \caption{Non-parametric test applied under independence (left) and dependence (right) assumptions. With dependence assumptions, the true pattern is recovered (group 1 below group 2) whereas the opposite is found with independence assumptions.}
         \label{fig:IndVsAgg}
     \end{subfigure}
\end{figure*}

Next, we show that when used as a measure of calibration error, the non-parametric test shows much less bias when testing under dependence and performs just as well as the benchmark $CE$ functions when testing under independence. To conduct this experiment, we extend the BBC framework in \cite{roelofs2022mitigating} to account for sample dependency. Intuitively, we make the extension such that individuals with more samples tend to have higher scores and are better calibrated. Specifically, suppose that there are two groups and our construction is as follows:
\begin{enumerate}
    \item \textbf{Per individual sample rates}: Each individual in group 1 has a random number of samples modeled by $N_{1}\sim 1+Pois(\lambda_{1})$ where $\lambda_{1}>0$. Each individual in group 2 has exactly one sample $N_{2}=1$.  
    \item \textbf{Score distribution}: The score distribution is the same for both groups, but for each individual is now sample-dependent and parameterized as $S(N) \sim Beta(max(N,N_{max}),1)$ where $N_{max}$ is a fixed constant.
    \item \textbf{Conditional expectations}: The conditional expectation function is also the same for both groups but is now also sample-dependent per individual. To model this, we pick $f(s)=g(s)$ in the parameterization $g^{-1}(b_{0}+b_{1}f(s))$ such that as $b_{0}\rightarrow 0$ and $b_{1} \rightarrow 1$, the closer we are to perfect calibration. Thus, a natural choice for $b_{0}$ and $b_{1}$ to have this behavior is to replace the terms with $b_{0}-zb_{0}$ and $b_{1}+z(1-b_{1})$ where $z=min(\frac{N}{N_{max}},1)$. 
\end{enumerate}

With this setup, we now benchmark the non-parametric test through the aforementioned bias estimation method  $Bias(f)=CE_{f}-TCE$. For the non-parametric test, $CE_{NP}$ is assessed as the $L_2$-norm of the distances between the estimated outcome and the score, that is $ CE_{NP}=\sqrt{\sum_{q=1}^{100}(\hat{f}(S(q))-S(q))^{2}}$ where $S(q)$ is the score at the $q$th quantile. For benchmarking, we will compare the the bias of the non-parametric estimator against that of the equal mass Expected Calibration Error (ECE, popularized in \cite{guoCalibration}) using 15 bins. However, several recent studies have identified issues with ECE in that using a fixed number of bins can introduce bias and additionally that is sensitive to sample size. We provide our own explanation of the asymptotic bias in Appendix \ref{appendix: ECE} and recommend the reader to \cite{Kumar2019,roelofs2022mitigating,Gruber2022,BlasiokCalibration} for further discussion of bias in ECE. Hence, we will also compare against the recently proposed Monotonic Sweep Calibration Error (MSCE) from \cite{roelofs2022mitigating}, which was designed to be an improved version of the ECE and does not use a fixed number of bins. As both of these are $L_{p}$ measures of calibration error, their biases are directly comparable to that of our non-parametric estimator. 

\begin{figure*}[ht]
    \centering
    \includegraphics[height=6cm]{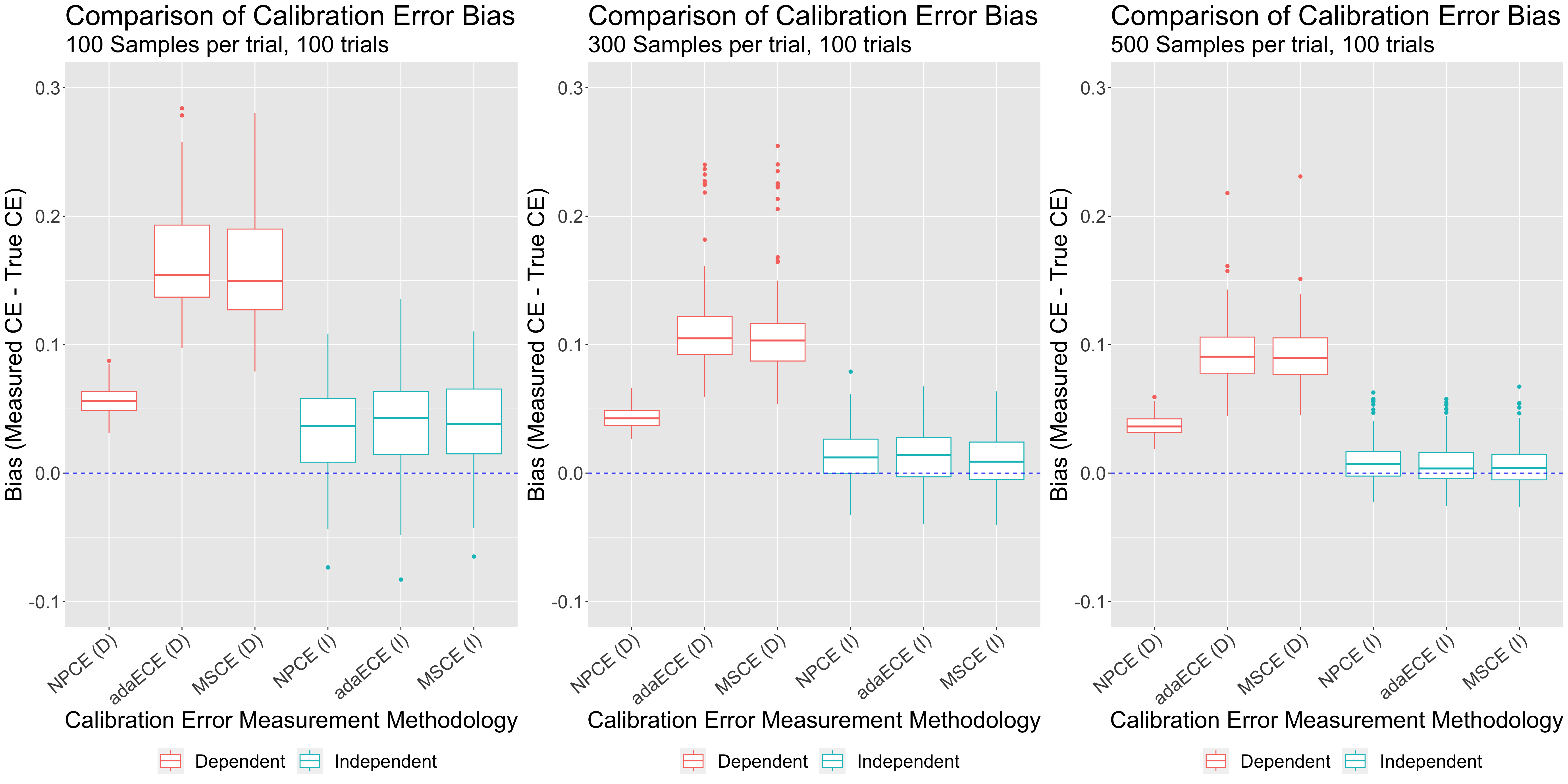}
    \caption{
    Non-parametric outcome test applied on dependently (red) vs. independently (blue) generated samples with bias measured on the Y-axis. We observe that $CE_{NP}$ is as good (in having low bias) as the benchmarks $CE_{ECE}$ and $CE_{MSCE}$ when dealing with independent samples. However, the non-parametric test is much better when it comes to measuring calibration error across dependent samples. }
    \label{fig:CeBiasComparison}
\end{figure*}

Figure ~\ref{fig:CeBiasComparison} shows the results of the benchmarking. Looking at the standard scenario of independent samples first, we first observe that the non-parametric test yields bias that is equally low compared to the benchmarks. Moreover, increasing sample size decreases the bias for all metrics unanimously, with 500 samples yielding near-zero bias. This trend corresponds to observations made in \cite{roelofs2022mitigating} and \cite{Gruber2022}. Pivoting over to the dependent samples, we observe that the non-parametric test has significantly lower bias compared to the benchmark and that again bias continues to shrink as sample size increases. The reason that the non-parametric test outperforms the benchmarks by a large margin in this case is that certain individuals have calibration errors that are significantly lower due to having fewer repeated samples. While the non-parametric test correctly captures these dependencies across samples, other calibration metrics do not. This is an important factor in real settings, as indeed models may be better calibrated for individuals with more samples. Measuring calibration in such settings while assuming independence may therefore lead to incorrect conclusions, as highlighted in our first experiment.

\section{Discussion}
\label{sec:discussion}
This work has proposed using non-parametric regression estimators to test for predictive parity and to provide a post-processing mitigation approach that calibrates models in cases where there is evidence of bias. We extended traditional results on the asymptotic normality of these estimators to a non-iid observation setting, allowing for valid statistical testing under the sorts of dependence that are common in industry applications.  We extend this methodology to address multi-objective models which are common in industry applications as well as to test for disparities in marginal outcomes.   

We illustrated the importance of distinguishing dependent from independent samples with our experiments. Not only can these assumptions lead to vastly different outcomes in fairness readouts, but testing for calibration under independence using metrics such as ECE or MSCE can can yield significant biases. 

We conclude by noting that there are numerous definitions of fairness \cite{Saravanakumar21}, which are often conflicting.  
The applicability of various fairness definitions has been the subject of numerous other works including \cite{simoiu2017} and \cite{corbett2018measure}. Predictive parity is an obvious choice when model scores need to be interpreted as likelihoods or risk scores, but encourage practitioners to use discretion when selecting an appropriate definition for their domain.  

\bibliographystyle{ACM-Reference-Format}
\bibliography{Fairness}

\pagebreak
\appendix

\section{Expected Calibration Error} \label{appendix: ECE}

The True Calibration Error metric is defined as
\[
\text{TCE} = \left( E \left( s - E(Y|s) \right)^2 \right)^{1/2}.
\]

The binning approximation to this metric is defined as the expected calibration error of the ``binned'' estimator:
\[
\bar s (s_i) := E( s | s \in B(s_i))
\]
where $B_1,...,B_K$ is an arbitrary binning of the support of the scores, and $B(s)$ equals the bin containing $s$.

A plug-in estimator of the binned TCE metric is given by 
\[
\widehat{ \text{TCE}}_{bin} := \left( \sum_{k=1}^K \frac{ \left| \left\{ s_i: s_i \in B_k  \right\} \right|}{n} \sum_{s \in B_k}  \left( \bar s_k - \bar y_k \right)^2 \right)^{1/2}
\]

where $\bar y_k = \frac{1}{ \left| \left\{ s_i: s_i \in B_k  \right\} \right|}\sum_{i: s_i \in B_k} y_i$ and $\bar y_k = \frac{1}{ \left| \left\{ s_i: s_i \in B_k  \right\} \right|}\sum_{s \in B_k} s_i$.
Since $\bar s_k \xrightarrow{a.s.} E(s | s \in B_k)$, $\bar y_k \xrightarrow{a.s.} E(y | s \in B_k)$, and  $\frac{\left| \left\{ s_i: s_i \in B_k  \right\} \right|}{n}\xrightarrow{a.s.} P(s \in B_k)$, it follow that
\[
\widehat{ \text{TCE}}_{bin} \xrightarrow{a.s.} \text{TCE}_{bin} := \sum_k P(s \in B_k) \left( E(y | s \in B_k) - E(s | s \in B_k)\right)^2
\]
which is the Expected Calibration Error of the binned estimator that predicts $E(s | s \in B_k)$ when $s \in B_k)$.  \cite{Kumar2019} note that $\text{TCE}_{bin} \leq \text{TCE}$, and we see that the empirical binned estimator is asymptotically biased.

\cite{roelofs2022mitigating} alternatively define a ``label-binned'' empirical calibration error,
\[
\widehat{ \text{TCE}}_{lb} := \sum_k \sum_{i: s_i \in B_k} \left(s_i - \bar y_k \right)^2
\]
Writing 
\begin{align*}
\left(s_i - \bar y_k \right)^2 & = \left(s_i - E(y_i|s_i \in B_k) \right)^2 + 2 \cdot \left(s_i - E(y_i|s_i \in B_k) \right) \\ 
& \cdot \left(\bar y_k - s_i \in B_k \right)  + \left(\bar y_k - s_i \in B_k \right)^2 
\end{align*}
it is readily seen that 
\[
\widehat{ \text{TCE}}_{lb} \xrightarrow{a.s.} E( s - E(Y| B(s)))^2
\]
where $B(s)$ denotes the bin containing $s$.  In general,  $E( s - E(Y| B(s)))^2 \neq \text{TCE}$ and we see that this estimator is also asymptotically biased.  Unlike the binned estimator, the label-binned variant does not appear to systemically over or under estimate the true calibration error.  As a simple example, suppose that scores are uniformly distributed on the interval $[0,1]$.  If $E(Y|s) = s$, the label-binned estimator asymptotically over estimates the $TCE$ when a single bin is used.  On the other hand, if $E(Y|s) = 1 - s$ the label-binned estimator asymptotically over estimates the $TCE$ when a single bin is used.

\subsection{Consistent estimation of the Expected Calibration Error}  

The inconsistency of the binned or label binned estimators stems from the inconsistency of the binned outcome values as an estimator of the conditional mean function.  As was seen in Section \ref{sec:mitigation}, this can be rectified by using an appropriate Nadaraya-Watson estimator.  Define 
\[
\widehat{\text{TCE}} := \frac{1}{n}\sum_i \left(s_i - \hat f(s_i) \right)^2
\]
where $\hat f(s)$ is the Nadaraya-Watson estimator of $E(Y|s)$.  This estimator is consistent for $\text{TCE}$.  To see this, write
\begin{align*}
\widehat{\text{TCE}} & = \frac{1}{n}\sum_i \left(s_i - \hat f(s_i) \right)^2 \\
& = \frac{1}{n}\sum_i \left(s_i -  f(s_i) \right)^2  + 2 \frac{1}{n}\sum_i \left(s_i -  f(s_i) \right) \left(f(s_i) - \hat f(s_i) \right) + \\
& \frac{1}{n}\sum_i \left(f(s_i) - \hat f(s_i) \right)^2~.
\end{align*}
The first term converges almost surely to the Expected Calibration Error.  The second two terms converge in probability to zero, which follows from uniform convergence of the Nadaraya-Watson estimator (as established, for instance, in \cite{luc78}).

\section{Multi-objective Toy Example}
\label{sec:mo_example}
Calibration of individual model scores may fail to provide calibration of the composite model when the various labels are correlated, particularly when the scores are correlated with prediction errors.  For instance, engagement metrics such as ``likes,'' ``comment,'' or ''click'' are often strongly correlated. As a simple toy example, consider models of two outcomes given by the following model. 
\[
Y^{(1)} = S^{(1)} + \epsilon^{(1)}, 
\]
and 
\[
Y^{(2)} = S^{(2)} + \epsilon^{(2)}~.
\]
Suppose that $S^(1)$ and $S^(1)$ are independent standard normal variables.  For simplicity of exposition, suppose that $\epsilon^(1) = S^{(2)}$ though the following argument holds whenever $\epsilon^{(1)}$ and $S^{(2)}$ are correlated.  Finally, suppose $\epsilon^{(2)}$ is standard normal and independent of all other quantities.  In this setting, $S^{(1)}$ and $S^{(2)}$ are calibrated predictors of $S^{(1)}$ and $S^{(2)}$, respectively.  Note that since $E\left(Y^{(1)} | S^{(1)} + S^{(2)} \right)=  S^{(1)} + S^{(2)}$ , and $E \left(Y^{(2)} | S^{(1)} + S^{(2)} \right) \neq  S^{(2)}$, it is readily seen that $S^{(1)} + S^{(2)}$ is not a calibrated predictor of $Y^{(1)} + Y^{(2)}$.  Furthermore, we have that
\[
E\left(Y^{(1)} + Y^{(2)} | \tilde S^{(1)}, \tilde S^{(2)} \right)=  S^{(1)} + 2 \cdot S^{(2)}
\]
for $\tilde S^{(i)} = E (Y^{(i)} | S^{(1)}, S^{(2)} )$.  Now, since $S^{(1)} + 2 \cdot S^{(2)}$ is a lower variance unbiased predictor of $Y^{(1)} + Y^{(2)}$ than $S^{(1)} + S^{(2)}$, we see that transforming the scores can improve mean-squared error of the overall predictions as well as the predictions of the individual models.  In general, if $E(Y^{(1)} | S^{(1)},..., S^{(K)}) \neq S^{(K)}$, this improvement in predictive performance is an expected consequence of the Gauss-Markov Theorem.

\section{Discussion of Calibration Approaches}\label{sec:calibration_appraoches}
In this section, we provide a brief review of alternate methods of calibrating a model, and give a brief comparison with the non-parametric approach.  A more thorough discussion is given by \cite{HuangTutoricalCalibration}.

\subsection{Binning}

Binning proceeds as follows.  Consider a binning $B_1,...,B_K$ of the score space with $B_i = [l_i, u_i)$.  For a given score, denote by $B(s)$ the bin $B_i$ containing $s$, i.e. for which $s \in B_i$. The binning approach defines transformed scores as $\tilde s = t(s, g)$ where 
\begin{align*}
t(s, g) & = \hat E (Y|S \in B(s), G = g) = \frac{\sum_{j} Y_j \cdot I\left\{ S_j \in B(s), G_j \in g \right\}}{\sum_{j} I\left\{ S_j \in B(s), G_j \in g \right\} }.
\end{align*}
A drawback of this approach is that it is subject to a bias-variance trade-off with respect to the choice of bins.  At one extreme, choosing a single bin amounts to replacing all scores with the average outcome for each group.  While this does calibrate the model, it is clearly unsatisfying with respect to model performance.  At the other extreme, choosing too many bins will result in many bins with a predicted average outcome of zero, which can fail to yield proper calibration asymptotically.  Both of these issues stem from the fact that, in general, we may have
\[
\hat E (Y|S \in B(s), G = g) \nrightarrow E (Y|S = s, G = g).
\]
Note that Theorem \ref{thm:asymptotic_normality} establishes the asymptotically optimal bin width.  In particular, choosing the Nadaray-Watson kernel to be the ``histogram'' kernel, namely
\[
K\left( x \right) = I \left\{ |x| < 1 \right\}~,
\]
we see that choosing the bandwidth parameter as $h = O(M^{-4/5})$ gives the best asymptotic tradeoff between bias and variance.  

\subsection{Isotonic regression:}

Isotonic regression estimates a non-decreasing conditional mean function.  Such as function can be identified by first solving the quadratic program

\[
\min_{\hat Y_1,...,\hat Y _M} \sum_i \left( \hat Y_i - Y_i \right)^2
\]
subject to the constraint that $\hat Y_i > \hat Y_j$ whenever $S_i> S_j$.  Then, any monotonic function $f(\cdot)$ satisfying $f(S_i) = \hat Y_i $ is an isotonic regression estimator.  In comparison with the Nadaraya-Watson estimators, these isotonic regression estimators are not robust against non-monotonicity of the conditional average outcomes.  

\subsection{Platt's scaling:}

Platt's scaling fits a logistic regression of the observed outcome labels on a binary classifier's scores, resulting in estimates of the observed outcome of the form
\[
\hat E \left( Y | S \right) = \frac{1}{1+ \exp (\beta_0 + \beta_1 S)}~.
\]

This methodology is generally preferable over non-parametric methods in small data setting, but is not robust against mispecification of the logistic function and is not applicable outside the binary classification setting.  Also, as with isotonic regression, this methodology assumes monotonicity of the conditional expected outcome, which may not hold in practice.  

\section{Alternate Formulation of predictive parity with Repeated Observations}\label{sec:alt_prp}

In repeated measurements setting, such as those commonly arising in internet applications where users repeatedly engage with a system (i.e. they may receive many notifications or see many ranked feed items arising from a recommender system), there are several choices for the conditional expectations which may be considered for assessing predictive parity which we will now discuss.  

Definition \ref{def:pred_rate_parity_dependent} asserts almost sure equality of $E(Y_{m,1} | S_{m,1} = s, G_{m,1} = g_1)$ and $E(Y_{m,1} | S_{m,1} = s, G_{m,1} = g_2)$. However, this quantity may differ from what is typically considered when assessing predictive parity.  Let $(Y^*_1, S^*_1, G^*_1),...,(Y^*_N, S^*_i ,G^*_N)$ for $N = \sum n_m$ be an arbitrary ordering of the observed data, $(Y_{m,i}, S_{m,i}, G_m)$.  An alternative formulation of predictive parity is to require almost sure equality of $E(Y^*_i | S^*_i = s, G^*_i = g_1)$ and $E(Y^*_i | S^*_i = s, G^*_i = g_2)$.

 In general, 
\[
E(Y_{m,1} | S_{m,1} = s, G_{m,1} = g) \neq E(Y^*_{i} | S^*_{i} = s, G^*_{i} = g),
\]
so the requirements of equality of the expectations cannot hold with respect to both formulations.  

We recommend considering equality of the conditional expectations of the form  $E(Y_{m,1} | S_{m,1} = s, G_{m,1} = g)$, which avoids giving undue influence to individuals with disproportionately high representation in the data. As a simple example, consider a health care algorithm predicting likelihood of having a disease.  A patient may repeatedly take such screening tests as part of a routine check-up until getting a positive outcome.  In such cases, affluent users with greater access to health care may tend to test more frequently and have disproportionately many negative outcomes.  This can downward-bias the estimates of average outcome for groups tending to have more affluent individuals.  In such cases, the interpretation of calibration is understood as a per-user quantity.  For instance, it is typically understood that 60\% of randomly chosen people who test with a score of .6 should have the disease, which is in agreement with the per-user formulation of predictive parity.

However, a practitioner who is interested in the quantities $E(Y^*_{i} | S^*_{i} = s, G^*_{i} = g)$ can estimate these expectations using the estimator
\[
\hat f_g(s) = \frac{\frac{1}{M}\sum_{m:G_m = g}\sum_{i = 1}^{n_m} Y_{m,i} K((S_{m,i} - s)/h)}{\frac{1}{M} \sum_{m:G_m = g} \sum_{i = 1}^{n_m} K((S_{m,i} - s)/h)}~.
\]
An analogous result to Theorem \ref{thm:asymptotic_normality} holds for this formulation. 

\section{Proof of Theorem \ref{thm:asymptotic_normality}}
We follow the proof of asymptotic normality for the Nadaraya-Watson estimator for the iid case in \cite{mcmurry2008} with needed modifications for dependency.  Note that the numerator of the Nadaraya-Watson estimator is an estimator of $a_g(s)$ while the denominator is an estimator of $b_g(s)$. 

\subsection{Supporting results}
We will begin by establishing the respective rates of convergence.  Define 
\[
\hat a_{g}(s) = \frac{1}{M}\sum_{m:G_m = g} \frac{1}{n_m}\sum_{i = 1}^{n_m} Y_{m,i} K\left ((S_{m,i} - s)/h \right)
\]
and 
\[
\hat b_{g}(s) = \frac{1}{M}\sum_{m:G_m = g} \frac{1}{n_m}\sum_{i = 1}^{n_m} K \left((S_{m,i} - s)/h \right)~.
\]
\begin{lemma}\label{lemma:meanConsistency}
Under Assumptions (A1)-(A5), (B1) and (B2),
\begin{align*}
E & \left( \frac{1}{M}\sum_{m:G_m = g} \frac{1}{n_m}\sum_{i = 1}^{n_m} Y_{m,i} K\left ((S_{m,i} - s)/h \right) \right) \\ 
& = P(G_m = g) \cdot a_g(s) +o (h^{d})
\end{align*}
and
\begin{align*}
E & \left(  \frac{1}{M}\sum_{m:G_m = g} \frac{1}{n_m}\sum_{i = 1}^{n_m} K \left((S_{m,i} - s)/h \right) \right) \\
& = P(G_m = g) \cdot b_g(s)+o (h^{d_1})    
\end{align*}
\end{lemma}

\begin{proof}[Proof of Lemma \ref{lemma:meanConsistency}]
Define $M_g = \left| \left\{ m: G_m = g\right\} \right|$ to be the number of members in group $g$. Note that 
\begin{align*}
E & \left( \frac{1}{M}\sum_{m:G_m = g}\frac{1}{n_m}\sum_{i = 1}^{n_m} Y_{m,i} K((S_{m,i} - s)/h) \right) \\
&= E \left( E \left( \left. \frac{1}{M}\sum_{m:G_m = g}\frac{1}{n_m}\sum_{i = 1}^{n_m} Y_{m,i} K((S_{m,i} - s)/h) \right|  n_m, M_g \right) \right) \\
&= E \left(\frac{1}{M}\sum_{m:G_m = g}\frac{1}{n_m}\sum_{i = 1}^{n_m}  E \left( \left. Y_{m,i} K((S_{m,i} - s)/h) \right| n_m, M_g \right) \right)\\
&= E \left( \frac{M_g}{M}Y_{m,1} K((S_{m,1} - s)/h)  \right) \\
&= P(G_m = g) E \left( Y_{m,1} K((S_{m,1} - s)/h) \right) \\
\end{align*}
Now, 
\begin{align*}
E & \left( Y_{m,1} K((S_{m,1} - s)/h) \right) = E \left( E\left( \left. Y_{m,1} K((S_{m,1} - s)/h) \right| S_{m,i}\right)\right) \\
&= E \left( E\left( \left. Y_{m,1} \right| S_{m,i}\right) K((S_{m,1} - s)/h) \right) \\
&= \int^{\infty}_{-\infty} E\left( \left. Y_{m,1} \right| S_{m,i} = s'\right) K((s' - s)/h) f_S(s') d s' \\
&= \int^{\infty}_{-\infty} f_{g}(s') K((s' - s)/h) f_S(s') d s' \\
\end{align*}

Fix an $\epsilon > 0$, and denote by $N_\epsilon (s)$ the neighborhood of radius epsilon centered around $s$.  Then, this integral can be expressed as 
\begin{align*}
\int^{\infty}_{-\infty} & f_{g}(s') K((s' - s)/h) f_S(s') d s' \\
&= \int_{N_\epsilon(s)} f_{g}(s') K((s' - s)/h) f_S(s') d s' \\ 
&+ \int_{N^c_\epsilon(s)} f_{g}(s') K((s' - s)/h) f_S(s') d s'   
\end{align*}

The assumption that $K$ has $d_k$ finite moments implies that 
\[
 \int_{N^c_\epsilon(s)} f_{g}(s') K((s' - s)/h) f_S(s') d s' = o(h^{d_k})
\]
Through a transformation of variables, $t = (s' - s)/h$
\begin{align*}
\int_{N_\epsilon(s)} & f_{g}(s') K((s' - s)/h) f_S(s') d s' \\
& = \int^{\epsilon/h}_{-\epsilon/h} f_{g}(s + ht) f_S(s + ht) K(t)  d t  
\end{align*}

By assumption, $f_{g}(s') f_S(s')$ has $d$ bounded and continuous derivatives on $N_\epsilon (s)$.  It follows that there exists an $s^* \in [s, s+ht]$ such that 
\begin{align*}
f_{g}(s + ht) f_S(s + ht) = & f_{g}(s) f_S(s) \sum_{d = 1}^{d_p-1}\frac{(ht)^d}{d!} (f_g \cdot f_S)^{(d)}(s) \\
&+ \frac{(ht)^{d_p}}{d_p!}(f_g \cdot f_S)^{(d_p)}(s + s^*) 
\end{align*}
and we can write 
\begin{align*}
& \int^{\epsilon/h}_{-\epsilon/h}  f_{g}(s + ht) f_S(s + ht) K(t)  d t   \\
& = \int^{\epsilon/h}_{-\epsilon/h} f_{g}(s) f_S(s)K(t)  d t +  \sum_{d' = 1}^{d-1} \int^{\epsilon/h}_{-\epsilon/h}\frac{(ht)^{d'}}{d'!} (f_g \cdot f_S)^{(d')}(s)K(t)  d t \\
& + \int^{\epsilon/h}_{-\epsilon/h}\frac{(ht)^{d}}{d!} (f_g \cdot f_S)^{(d)}(s + s^*) K(t)  d t
\end{align*}

By the assumption on the tail decay of $K(\cdot)$, and the fact that $K(\cdot)$ integrates to one, we have that 
\begin{align*}
\int^{\epsilon/h}_{-\epsilon/h} f_{g}(s + ht) f_S(s + ht) K(t)  d t 
&= \int^{\infty}_{-\infty} f_{g}(s) f_S(s) K(t)  d t + o(h^{d}) \\
& = b_g(s) + o(h^{d}). \\ 
\end{align*}

Similarly, we have 
\[
\int^{\epsilon/h}_{-\epsilon/h} \frac{(ht)^{d'}}{d'!} (f_g \cdot f_S)^{(d')}(s)K(t)  d t = o(h^{d})
\]
for $d' = 1,...,d-1$ since the first $d$ moments are assumed to be zero.  Finally, since the derivatives of $(f_g \cdot f_S(s))$ are bounded and continuous, we have that there exists a $\delta >0$ such that $(f_g \cdot f_S)^{(d_p)}(s+s^*) - (f_g \cdot f_S)^{(d_p)}(s) < \delta$ for all $h$ sufficiently small. Therefore
\begin{align*}
\int^{\epsilon/h}_{-\epsilon/h} & \frac{(ht)^{d_p}}{d_p!} (f_g \cdot f_S)^{(d)}(s+ s^*)K(t)  d t \\
& = \int^{\epsilon/h}_{-\epsilon/h}\frac{(ht)^{d_p}}{d_p!} (f_g \cdot f_S)^{(d_p)}(s)K(t)  d t  \\
& +  \int^{\epsilon/h}_{-\epsilon/h}\frac{(ht)^{d_p}}{d_p!} \left( (f_g \cdot f_S)^{(d_p)}(s+s^*)  - (f_g \cdot f_S)^{(d_p)}(s) \right) K(t) dt \\
& + o(h^{d_1})\\
& = o(h^{d_1})
\end{align*}
since 
\begin{align*}
\int^{\epsilon/h}_{-\epsilon/h} & \frac{(ht)^{d_p}}{d_p!} \left( (f_g \cdot f_S)^{(d_p)}(s+s^*)  - (f_g \cdot f_S)^{(d_p)}(s) \right)K(t) dt \\
& \leq \int^{\epsilon/h}_{-\epsilon/h}\frac{(ht)^{d_p}}{d_p!} \delta K(t) dt
\end{align*}
for $h$ suitably small, and the latter term is  $o(h^{d_1})$.  This proves the first convergence in the lemma.  The second can be proven with a similar, albeit less involved argument. 
\end{proof}

\begin{lemma}\label{lemma:varConsistency}
Under Assumptions (A1)-(A5), (B1) and (B2),
\begin{align*}
\text{var} & \left( \frac{1}{M}\sum_{m:G_m = g} \frac{1}{n_m}\sum_{i = 1}^{n_m} Y_{m,i} K((S_{m,i} - s)/h) \right) \\
=& \frac{\left( \overline{f^2_g}(s) + \overline{\sigma^2_g}(s ) \right) f_S(s) + \overline{\rho_g}(s; n_m)  f^2_S(s)}{Mh} \cdot \int^\infty_{-\infty} K^2(t) dt \\ 
&  + O(M^{-1}) +o((Mh)^{-1}),
\end{align*}
\begin{align*}
\text{var} & \left( \frac{1}{M}\sum_{m:G_m = g} \frac{1}{n_m}\sum_{i = 1}^{n_m} K((S_{m,i} - s)/h) \right) \\
=& \frac{\left( \overline{e^2_g}(s) + \overline{v^2_g}(s ) \right) f_S(s)\int^\infty_{-\infty}\cdot K^2(t) dt}{Mh} + O(M^{-1}) +o((Mh)^{-1}),
\end{align*}
and 
\begin{align*}
\text{cov} & \left( \frac{1}{M}\sum_{m:G_m = g} \frac{1}{n_m} \sum_{i = 1}^{n_m} Y_{m,i} K((S_{m,i} - s)/h), \right.  \\
& \left. \frac{1}{M}\sum_{m:G_m = g} \frac{1}{n_m}\sum_{i = 1}^{n_m} K((S_{m,i} - s)/h) \right)  \\
 = & \frac{ \overline{f_g}(s)  f_S(s)\int^\infty_{-\infty}\cdot K^2(t) dt}{Mh} + O(N^{-1}) +o((Mh)^{-1}),
\end{align*}
\end{lemma}

\begin{proof}[Proof of Lemma \ref{lemma:varConsistency}]
We begin by establishing the convergence of
\begin{align*}
\text{var} &\left( \frac{1}{M}\sum_{m} \frac{1}{n_m}\sum_{i = 1}^{n_m} Y_{m,i,g} K((S_{m,i} - s)/h) \right) \\
& = \frac{1}{M^2} \sum_m \text{var} \left( \frac{1}{n_m}\sum_{i = 1}^{n_m} Y_{m,i,g} K((S_{m,i} - s)/h) \right)~.
\end{align*}
By the law of total variances, 
\begin{align*}
\text{var} & \left( \frac{1}{n_m}\sum_{i = 1}^{n_m} Y_{m,i,g} K((S_{m,i} - s)/h) \right) \\
= & \text{var} \left( E\left( \left. \frac{1}{n_m}\sum_{i = 1}^{n_m} Y_{m,i,g} K((S_{m,i} - s)/h)\right| n_m\right) \right) \\
&+ E \left( \text{var} \left( \left. \frac{1}{n_m}\sum_{i = 1}^{n_m} Y_{m,i,g} K((S_{m,i} - s)/h)\right| n_m\right) \right) 
\end{align*}
 
It is readily seen from the equations derived below that the first term on the right hand side is constant up to a $o(h^d)$ term.  The inner variance in the second term can be written as
\begin{align*}
\text{var} &\left( \left. \frac{1}{n_m}\sum_{i = 1}^{n_m} Y_{m,i,g} K((S_{m,i} - s)/h)\right| n_m\right) \\
& = \frac{1}{n^2_m}\sum_{i = 1}^{n_m} \text{var} \left( \left.  Y_{m,i,g} K((S_{m,i} - s)/h)\right|n_m \right)\\
& + \frac{1}{n^2_m} \sum_{i \neq j} \text{cov} \left( \left.  Y_{m,i,g} K((S_{m,i} - s)/h), Y_{m,j,g} K((S_{m,j} - s)/h)\right|n_m \right)
\end{align*}

Again applying the law of total variance, we have
\begin{align*}
\text{var} & \left( \left. Y_{m,i,g} K((S_{m,i} - s)/h)\right|n_m \right) \\
& =  \text{var} \left( \left.  E \left( \left. Y_{m,i,g} K((S_{m,i} - s)/h) \right|n_m, S_{m,i} \right) \right|n_m \right) \\
&+ E \left( \left.  \text{var} \left( \left. Y_{m,i,g} K((S_{m,i} - s)/h) \right|n_m, S_{m,i} \right) \right|n_m \right) 
\end{align*}

Note that 
\begin{align*}
E & \left( \left. Y_{m,i,g} K((S_{m,i} - s)/h) \right|n_m, S_{m,i} \right) \\
& = f_g(S_i; n_m) \cdot K((S_{m,i} - s)/h) 
\end{align*}
and 
\begin{align*}
\text{var} & \left( \left. Y_{m,i,g} K((S_{m,i} - s)/h) \right|n_m, S_{m,i} \right) \\
= & K^2((S_{m,i} - s)/h) \text{var} \left( \left. Y_{m,i,g}  \right|n_m, S_{m,i} \right) \\
= & K^2((S_{m,i} - s)/h) \sigma^2_g(S_{m,i};n_m).
\end{align*}

Arguing as in Lemma \ref{lemma:meanConsistency}, it is readily seen that 
\begin{align*}
E & \left( \left. f_g(S_i; n_m) \cdot K((S_{m,i} - s)/h) \right| n_m \right) = f_g(s; n_m) + o(h^{d})    
\end{align*}

It follows that 
\begin{align*}
\text{var} & \left( \left.  Y_{m,i,g} K((S_{m,i} - s)/h)\right|n_m \right) \\
= & E \left( \left. \left[  f^2_g(S_i; n_m) + \sigma^2_g(S_{m,i};n_m) \right] \cdot K^2((S_{m,i} - s)/h)  \right| n_m \right) \\
& + f^2_g(s; n_m)  + o(h^{d_1}) \\
=& \int^\infty_{-\infty} \left[  f^2_g(s'; n_m) + \sigma^2_g(s';n_m) \right] \cdot f_S(s')  \cdot K^2( s' - s)/h) ds' \\
& +f^2_g(s; n_m)  + o(h^{d_1}) \\
= & \frac{1}{h} \int^\infty_{-\infty} \left[  f^2_g(s + th; n_m) + \sigma^2_g(s + th;n_m) \right] \cdot f_S(s + th) \cdot K^2(t) dt \\
&+f^2_g(s; n_m)  + o(h^{d_1})\\
= & \frac{ \left(  f^2_g(s; n_m) + \sigma^2_g(s ;n_m) \right) \cdot f_S(s) }{ \int^\infty_{-\infty} K^2(t) dt} +f^2_g(s; n_m)  + o(h^{d})
\end{align*}
To address the covariance terms, we first apply the law of total covariance, which gives 
\begin{align*}
 \text{cov} & \left( \left. Y_{m,i,g} K((S_{m,i} - s)/h),  Y_{m,j,g} K((S_{m,j} - s)/h)\right| n_m \right) \\
 & = E \left( \text{cov} \left( Y_{m,i,g} K((S_{m,i} - s)/h), \right. \right. \\
 & \left. \left. \left. Y_{m,j,g} K((S_{m,j} - s)/h)\right|n_m, S_{m,i}, S_{m,j}\right) \right) \\
 & + \text{cov} \left( E \left( \left. Y_{m,i,g} K((S_{m,i} - s)/h)\right|n_m, S_{m,i}, S_{m,j}\right),  \right. \\
 & \left. E \left( \left.  Y_{m,j,g} K((S_{m,j} - s)/h)\right|n_m, S_{m,i}, S_{m,j}\right) \right)
\end{align*}
Using the independence specified in assumption (B5), we have that
\begin{align*}
E & \left( \left.  Y_{m,i,g} K((S_{m,i} - s)/h)\right|n_m, S_{m,i}, S_{m,j}\right) \\
& = E \left( \left.  Y_{m,i,g} K((S_{m,i} - s)/h)\right|n_m, S_{m,i}\right) \\
& = f_g(S_{m,i}; n_m) \cdot K((S_{m,i} - s)/h)~.
\end{align*}
Again using the independence, we have that the covariance of these terms is zero.  When this independence does not hold, the covariance of these terms depend on the joint distribution of the scores and labels, and possibly the index when assumption (B5) is not met. 
Now, 
\begin{align*}
\text{cov} & \left( \left.  Y_{m,i,g} K((S_{m,i} - s)/h), Y_{m,j,g} K((S_{m,j} - s)/h)\right|n_m, S_{m,i}, S_{m,j}\right) \\
= & \text{cov} \left( \left.  Y_{m,i,g} , Y_{m,j,g} \right|n_m, S_{m,i}, S_{m,j}\right) \\
& \cdot K((S_{m,i} - s)/h) \cdot K((S_{m,j} - s)/h) \\
=& \rho_g \left( S_{m,i}, S_{m,j}; n_m \right)\cdot K((S_{m,i} - s)/h) \cdot K((S_{m,j} - s)/h) 
\end{align*}
The expectation of this quantity is 
\begin{align*}
E & \left(\rho_g \left( S_{m,i}, S_{m,j}; n_m \right)\cdot K((S_{m,i} - s)/h) \cdot K((S_{m,j} - s)/h)  \right) \\
=& \int \int \rho_g \left( s_1, s_2; n_m \right)\cdot K((s_1 - s)/h) \\
& \cdot K((s_2 - s)/h) f_S(s_1)f_S(s_2)ds_1 ds_2 \\
= & \frac{\rho_g \left( s, s; n_m \right) f^2_S(s_1)}{h}\int^\infty_{-\infty} K^2(t) dt + o(h^d)~
\end{align*}
We have made use of assumption (B5) in the first equality.  If this does not hold, the computation requires use of the joint distribution of two scores.  If assumption (B4) does not hold, this further requires the joint distribution to depend on the indices of the scores.   

Therefore, 
\begin{align*}
\text{var} &\left( \frac{1}{n_m}\sum_{i = 1}^{n_m} Y_{m,i,g} K((S_{m,i} - s)/h) \right)
\\
= & \frac{1}{h}E\left( \frac{1}{n_m}  f^2_g(s; n_m) + \sigma^2_g(s ;n_m) \right) f_S(s)\int^\infty_{-\infty}\cdot K^2(t) dt \\
& +\frac{1}{h}E\left( \frac{n_m(n_m-1)}{n_m}  \rho_g(s,s; n_m) ) \right) f^2_S(s)\int^\infty_{-\infty}\cdot K^2(t) dt \\
& + E\left( f^2_g(s; n_m) \right) +   + o(h^{d_1})\\
 =& \frac{1}{h}\left( \left( \overline{f^2_g}(s) + \overline{\sigma^2_g}(s ) \right) f_S(s) + \overline{\rho_g}(s; n_m)  f^2_S(s)\right)  \\
 &+  E\left( f^2_g(s; n_m) \right)  + o(h^{d})
\end{align*}
Combining these equations, we have that
\begin{align*}
\frac{1}{M^2} \sum_m \text{var}& \left( \frac{1}{n_m}\sum_{i = 1}^{n_m} Y_{m,i,g} K((S_{m,i} - s)/h) \right)\\
= &  \frac{\left( \overline{f^2_g}(s) + \overline{\sigma^2_g}(s ) \right) f_S(s) + \overline{\rho_g}(s; n_m)  f^2_S(s)}{Mh} \int^\infty_{-\infty}\cdot K^2(t) dt \\
& + O(M^{-1}) +o((Mh)^{-1}),
\end{align*}
which proves the first convergence in the lemma.  The other convergence results are proved by similar arguments.  
\end{proof}

We now show the joint asymptotic normality of $a_g(s)$ and $b_g(s)$

\begin{lemma}\label{lemma:joint}
Under Assumptions (A1)-(A5), (B1) and (B2),
\begin{align*}
\sqrt{Mh} \left( c_1 \cdot \hat a_g(s) + \right. &\left. c_2 \cdot \hat b_g(s) - c_1 \cdot a_g(s) + c_2 \cdot b_g(s) \right)  \rightarrow N \left( 0, \gamma(s) \right)
\end{align*}
where $\gamma(s)$ is the appropriate variance according to the variances and covariances given in Lemma \ref{lemma:varConsistency}.
\end{lemma}
\begin{proof}[Proof of Lemma \ref{lemma:joint}]
This Lemma follows immediately from the Liapunov Central Limit Theorem where the Liapunov condition follows from arguing as in \ref{lemma:meanConsistency} and from assumption 6 on the finite moments of the outcome variable. 
\end{proof}

\subsection{Main Proof}
The main asymptotic normality result is a consequence of this Lemma.  
\begin{proof}[Proof of Theorem \ref{thm:asymptotic_normality}]
Arguing as in Lemma 6 of \cite{mcmurry2008}, it is readily seen that there exists a constant $c > 0$ such that $\hat f_g(s) > c$ for all $n$ sufficiently large.  Therefore, it follows from the intermediate value theorem and Lemma \ref{lemma:meanConsistency} that 
\begin{align*}
\hat f_g(s) -  f_g(s) =& \hat a_g(s) \left( \frac{1}{b_g(s)} - \frac{1}{\delta_n^2} (\hat b_g (s) - b_g(s)) \right) -\frac{a_g(s)}{b_g(s)}\\
= & \frac{1}{b_g(s)}(\hat a_g(s)  - a_g(s) ) + \frac{\hat a_g(s)}{\delta_n^2}(\hat b_g(s)  - b_g(s) ) + o(h^d)
\end{align*}
where 
\[
|\delta_n - b_g(s)| \leq |\hat b_g(s) - b_g(s)|~.
\]
Note that because $\hat b_g(s)$ converges to $b_g(s)$, we have that $\delta_n$ also converges to $b_g(s)$. It follows from Slutsky's Theorem and Lemma \ref{lemma:meanConsistency} that the asymptotic normality in Theorem \ref{thm:asymptotic_normality} holds with 
\begin{align*}
\sigma^2_g (x) = &   \frac{1}{b_g^2(s)P(G = g)^2}  \cdot \left( \overline{f^2_g}(s) + \overline{\sigma^2_g}(s ) \right) f_S(s)\int^\infty_{-\infty}\cdot K^2(t) dt  \\ 
&  + \frac{a_g^2(s)}{b_g^4(s)P(G = g)^2}  \cdot \left( \overline{e^2_g}(s) + \overline{v^2_g}(s ) \right) f_S(s)\int^\infty_{-\infty}\cdot K^2(t) dt \\
&  - 2 \frac{a_g(s)}{b_g^3(s)P(G = g)^2} \cdot \overline{f_g}(s)  f_S(s)\int^\infty_{-\infty}\cdot K^2(t) dt
\end{align*}
\end{proof}

\section{Proof of Other Results}
\begin{proof}[Proof of Lemma \ref{lemma:multi-fairness}]
\begin{align*}
    E & \left( Y_i | S_i = s \right) = E\left( \sum_{k=1}^K w_k Y^k_i |\sum_{k=1}^K w_k S^k_i = s\right) \\
    &= E \left( E\left(\left.  \sum_{k=1}^K w_k Y^k_i \right| S^1_i = s_1,...,S^K_i =s_k, \right. \right. \\ 
    & \qquad \qquad \left. \left. \left. \sum_{k=1}^K w_k S^k_i = s \right) \right|  \sum_{k=1}^K w_k S^k_i = s \right) \\
    &= E \left( \sum_{k=1}^K w_k E\left( \left. Y^k_i \right| S^1_i = s_1,...,S^K_i =s_k, \right. \right. \\
     & \qquad \qquad \left. \left. \left.  \sum_{k=1}^K w_k S^k_i = s \right) \right| \sum_{k=1}^K w_k S^k_i = s \right)  \\
    &= E \left( \sum_{k=1}^K w_k s_k | \sum_{k=1}^K w_k S^k_i = s \right) \\
   & = s
\end{align*}
\end{proof}

\begin{proof}[Proof of Theorem \ref{thm:multi-fairness}]
The result follows immediately from the fact that the transformed scores satisfy
\[
E(Y^k | \tilde S_1 = s_1,..., \tilde S_K = s_K, G= g_i) = s_k
\]
for all $s_1,...,s_K$ and $g_i$.
\end{proof}
\begin{proof}[Proof of Theorem \ref{thm:inframarginality}]
The proof readily follows that of Theorem \ref{thm:asymptotic_normality}. 
\end{proof}
\begin{proof}[Proof of Theorem \ref{thm:infra2}]
Let $\tilde S_g$ denote calibrated scores from group $g$.  Let $\tilde f_{g} (\cdot)$ and $\tilde F_{g} (\cdot)$ denote the conditional average outcome function and distribution function of the calibrated scores, respectively.  The system of equations specified in the Theorem now reduces to finding $t'$ satisfying the single constraint that
\[
\sum_i p_{g_i} \int_0^{t'} \tilde f_{g_1} (t) \tilde F_{g_i}(t) =  \sum_i p_{g_i} \int_0^{t^*} \hat f_{g_1} (t) \hat F_{g_i}(t)
\]
Because 
\[
\sum_i p_{g_i} \int_0^{0} \tilde f_{g_1} (t) \tilde F_{g_i}(t) = 0
\]
and 
\[
\sum_i p_{g_i} \int_0^{\infty} \tilde f_{g_1} (t) \tilde F_{g_i}(t) > \sum_i p_{g_i} \int_0^{t^*} \hat f_{g_1} (t) \hat F_{g_i}(t)
\]
the result follows immediately from the intermediate value theorem. 
\end{proof}

\begin{proof}[Sketch of Proof of Theorem \ref{thm:delta_classification}] 
\begin{align*}
\widehat{\Delta-PCC}_g & = \frac{1}{n_g}\sum_{i:G_i = g} I \left\{  S_i > t\right\}  -  I \left\{  \hat f_g(S_i) > t\right\} \\
&= \frac{1}{n_g}\sum_{i:G_i = g} I \left\{  S_i > t\right\}  -  I \left\{  \hat S^*_i > t\right\} \\
& + \frac{1}{n_g}\sum_{i:G_i = g} I \left\{  S^*_i > t\right\}  - I \left\{  \hat f_g(S_i) > t\right\}
\end{align*}

The first term on the right hand side converges in probability to $\Delta-PCC_g$.  To see that the second term converges in probability to zero, we have that for any $\epsilon > 0$,

\begin{align*}
P & \left( \left| \frac{1}{n_g}\sum_{i:G_i = g} I \left\{  S^*_i > t\right\}  - I \left\{  \hat f_g(S_i) > t\right\} \right| > \epsilon \right) \\
& \leq \frac{1}{\epsilon^2} E \left( \left| \frac{1}{n_g}\sum_{i:G_i = g} I \left\{  S^*_i > t\right\}  - I \left\{  \hat f_g(S_i) > t\right\} \right| \right) \\
& \leq \frac{1}{\epsilon^2} E \left( \left| I \left\{  S^*_1 > t\right\}  - I \left\{  \hat f_g(S_1) > t\right\} \right| \right) \\
& \leq \frac{1}{\epsilon^2} P \left( \left| S^*_1 - \hat f_g(S_1) \right| > |t-S^*| \right) \\ 
& \leq \frac{1}{\epsilon^2} E \left( P \left(\left. \left| S^*_1 - \hat f_g(S_1) \right| > |t-S^*| \right| S^*  \right) \right)~. 
\end{align*}
The inner probability tends to zero uniformly in $S^*$ because of the uniform convergence of $\hat f_g(\cdot)$ (e.g. see \cite{luc78}).  The outer expectation then tends to zero as a consequence of the dominated convergence theorem.  
\end{proof}

\begin{proof}[Sketch of Proof of Theorem \ref{thm:delta_rank}] 
Writing the quantities of 
\[R(S_{i,q},q) - R(\hat f_g(S_{i,q}),q)\] 
as indicators that the rank statistics of the $\hat f_g(S_{i,q})$ equal those of the $S^*_{i,q}$, the convergence follows from a similar argument to Theorem \ref{thm:delta_classification}.
\end{proof}
\end{document}